\documentclass[11pt,a4paper]{article}
\newcommand{\email}{\texttt}

\DeclareSymbolFont{AMSb}{U}{msb}{m}{n}
\DeclareSymbolFontAlphabet{\mathbb}{AMSb}
\newcommand{\binom}[2]{{#1 \choose #2}}

\newtheorem{theorem}{Theorem}[section]
\newtheorem{proposition}[theorem]{Proposition}

\newtheorem{lemma}[theorem]{Lemma}

\newtheorem{definition}[theorem]{Definition}
\newtheorem{example2}[theorem]{Example}

\newcommand{\qed}{}
\newenvironment{proof}{\begin{genproof}}{\end{genproof}}

\newenvironment{genproof}[1][]{\begin{trivlist}\item \textbf{Proof#1:} }{\nolinebreak\qquad\nolinebreak\framebox(5,5)[lb]{}\nolinebreak\end{trivlist}}

\newenvironment{genproofnoqed}[1][]{\begin{trivlist}\item \textbf{Proof#1:} }{\nolinebreak\end{trivlist}}
\newenvironment{equation*}{$$}{$$}
\newenvironment{acknowledgements}{\subsection*{Acknowledgements}}{}

\usepackage[font=small,labelfont=bf]{caption}
\usepackage{graphicx}
\usepackage{url}
\usepackage{a4wide}
\usepackage{tikz}
\usetikzlibrary{positioning,automata}
\usetikzlibrary{calc}
\usetikzlibrary{decorations.markings}

\begin{document}

\title{Parameterized Complexity of Induced Graph Matching on Claw-Free Graphs\thanks{An extended abstract of the results in this paper have appeared in the Proceedings of 20th European Symposium on Algorithms (ESA 2012)~\cite{HermelinML12}. Partially supported by ERC StG project PAAl no.\ 259515.}}

\author{Danny Hermelin\footnote{Ben-Gurion University of the Negev, Beer-Sheva, Israel, \email{hermelin@bgu.ac.il}.} \and Matthias Mnich\footnote{Cluster of Excellence MMCI, Saarbr{\"u}cken, Germany, \email{mmnich@mmci.uni-saarland.de}} \and Erik~Jan~van~Leeuwen\footnote{Max-Planck Institut f\"{u}r Informatik, Saarbr\"{u}cken, Germany, \email{erikjan@mpi-inf.mpg.de}}}
\date{}

\maketitle

\begin{abstract}
The {\sc Induced Graph Matching} problem asks to find $k$ disjoint induced subgraphs isomorphic to a given graph~$H$ in a given graph $G$ such that there are no edges between vertices of different subgraphs.
This problem generalizes the classical {\sc Independent Set} and {\sc Induced Matching} problems, among several other problems.
We show that {\sc Induced Graph Matching} is fixed-parameter tractable in $k$ on claw-free graphs when $H$ is a fixed connected graph, and even admits a polynomial kernel when~$H$ is a complete graph.
Both results rely on a new, strong, and generic algorithmic structure theorem for claw-free graphs.

Complementing the above positive results, we prove $\mathsf{W}[1]$-hardness of {\sc Induced Graph Matching} on graphs excluding $K_{1,4}$ as an induced subgraph, for any fixed complete graph $H$.
In particular, we show that {\sc Independent Set} is $\mathsf{W}[1]$-hard on $K_{1,4}$-free graphs.

Finally, we consider the complexity of {\sc Induced Graph Matching} on a large subclass of claw-free graphs, namely on proper circular-arc graphs. We show that the problem is either polynomial-time solvable or $\mathsf{NP}$-complete, depending on the connectivity of $H$ and the structure of $G$.
\end{abstract}

\section{Introduction}
A graph is \emph{claw-free} if no vertex in the graph has three pairwise nonadjacent neighbors, i.e.~if it does not contain a copy of $K_{1,3}$ as an induced subgraph.
The class of claw-free graphs contains several well-studied graph classes such as line graphs, unit interval graphs, de Bruijn graphs, the complements of triangle-free graphs, and graphs of several polyhedra and polytopes.
Consequently, claw-free graphs have attracted much interest, and are by now the subject of numerous research papers (see, e.g., the surveys by Faudree et al.~\cite{FaudreeFR97} and Chudnovsky and Seymour~\cite{ChudnovskySeymour05}).

Our understanding of claw-free graphs and their structure was greatly extended with the recently developed theory of Chudnovsky and Seymour.
This highly technical and detailed claw-free structure theory is contained in a sequence of seven papers~\cite{ChudnovskyS2007,ChudnovskyS2008a,ChudnovskyS2008b,ChudnovskyS2008c,ChudnovskyS2008d,ChudnovskyS2010,ChudnovskyS2012} (see also the accessible survey of Chudnovsky and Seymour~\cite{ChudnovskySeymour05}), and culminates in several variants of \emph{claw-free decomposition theorems}.
Each of these decomposition theorems shows that every connected claw-free graph can be constructed by starting with a collection of ``basic claw-free graphs'', and then gluing these basic graphs together in some controlled manner.

Recently, Chudnovsky and Seymour's claw-free structure theory has been deployed for algorithmic purposes.
For example, it was used to develop approximation algorithms for {\sc Chromatic Number} on quasi-line graphs~\cite{ChudnovskyOvetsky07, KingReed08}.
Faenza et al.~\cite{FaenzaOrioloStauffer11} gave an algorithm to find a decomposition of claw-free graphs, and used this algorithm to expedite the fastest previously known algorithm for {\sc Weighted Independent Set} on claw-free graphs, from $O(n^6)$ to $O(n^3)$.
Their decomposition, although similar in nature, does not rely on Chudnovsky and Seymour's structure theorem.
In contrast, the authors of this paper, along with Woeginger, independently gave an algorithm for the Chudnovsky-Seymour decomposition theorem, and used this to show that {\sc Dominating Set} is fixed-parameter tractable on claw-free graphs, and admits a polynomial kernel as well~\cite{HermelinMLW11}.
(The fixed-parameter tractability result was shown independently by Cygan et al.~\cite{CyganPPPW11} using different methods.)
Recently, Golovach et al.~\cite{GolovachPL12} used the algorithmic decomposition theorem by Hermelin et al.~\cite{HermelinMLW11} to give a fixed-parameter algorithm for the {\sc Induced Disjoint Paths} problem on claw-free graphs.

In this paper, we extend this line of research by considering the {\sc Induced Graph Matching} problem on claw-free graphs.
This problem generalizes several important and well-studied problems, such as {\sc Independent Set} and {\sc Induced Matching}.
It can be defined as follows.
We are given (claw-free) graphs $G,H$ and an integer $k$, and the goal is to determine whether there is an induced occurrence of $k \cdot H$ in $G$; that is, whether there is a set of induced subgraphs $M = \{H_1,\ldots,H_k\}$ in $G$, each isomorphic to~$H$, that are pairwise vertex-disjoint and satisfy that $\{u,v\} \notin E(G)$ for any pair of vertices $u \in V(H_i)$ and $v \in V(H_j)$ with $i \neq j$.
Such an induced occurrence of $k\cdot H$ is called an \emph{induced $H$-matching} in $G$ of size $k$.
As discussed further below, the study of this problem among other things requires a significantly stronger decomposition theorem for claw-free graphs compared to those developed previously.

The {\sc Induced Graph Matching} problem is closely related to the equally well-studied {\sc Graph Matching} problem.
In {\sc Graph Matching}, the goal is to find an \emph{$H$-matching}, that is, a set of subgraphs $H_{1},\ldots,H_{k}$ of $G$ that are each isomorphic to $H$ and pairwise vertex-disjoint (but edges between copies of $H$ are allowed).
Observe that when $H$ has minimum degree at least 3, {\sc Graph Matching} reduces to {\sc Induced Graph Matching} by subdividing all edges of $G$ and $H$.
On the other hand, on line graphs (a subclass of claw-free graphs), a reduction exists in the other direction.
Recall that the \emph{line graph} $L(G)$ of a graph $G$ consists of a vertex $v_{e}$ for each $e \in E(G)$ and there is an edge $\{v_{e},v_{f}\}$ in $L(G)$ if and only if $e$ and $f$ are incident to the same vertex in $G$.
The graph $G$ is known as the \emph{pre-image} of $L(G)$.
It can be easily seen that there is a bijection between the set of induced $L(H)$-matchings on $L(G)$ and the set of $H$-matchings on $G$ when $L(H)$ is not a triangle.
Hence, {\sc Induced Graph Matching} on line graphs inherits essentially all known complexity results of {\sc Graph Matching} on general graphs.
In particular, this implies that {\sc Induced Graph Matching} is $\mathsf{NP}$-complete on line graphs if $H$ is not edgeless~\cite{GareyJ1979,KirkpatrickH1983,KoblerR2003} and not a triangle\footnote{Observe that the reduction given here does not work in the case that $H$ is a triangle, because the pre-image of a triangle is not unique.
The triangle is also the only line graph for which the pre-image is not unique~\cite{Whitney1932}, which is why the reduction works for all other cases.
We give a proof of the $\mathsf{NP}$-hardness of {\sc Induced Graph Matching} for $H = K_3$ in Appendix~\ref{sec:k3}.}.

In the context of parameterized complexity, several results on {\sc Graph Matching} are known.
In particular, the problem is known to be fixed-parameter tractable by the size of the matching~$k$~\cite{FellowsKNRRSTW2008,KneisMRR2006} and even has a polynomial kernel (when $H$ is fixed)~\cite{Moser2009}.
Recently, tight lower bounds on the kernel size for specific graphs $H$ were obtained by Dell and Marx~\cite{DellMarx2012}, and Hermelin and Wu~\cite{HermelinWu2012}.
Note again that these lower bound results immediately carry over to {\sc Induced Graph Matching} on line graphs when $H$ is not a triangle.

On general graphs, the {\sc Induced Graph Matching} problem is $\mathsf{W}[1]$-hard for any complete graph~$H$ when parameterized by the matching size $k$~\cite{DowneyFellows1999,MoserThilikos2006}.
Marx~\cite{Marx2005} showed that the {\sc Induced Graph Matching} problem for a graph $H$ on a single vertex (i.e.~the {\sc Independent Set} problem) is $\mathsf{W}[1]$-hard on $K_{1,5}$-free graphs.
Another related result by Cameron and Hell~\cite{CameronHell2006} shows that on certain graph classes the problem can be reduced to an instance of {\sc Independent Set} on a graph in that same graph class, provided that the set of all occurrences of $H$ are given.
Their results, however, do not provide methods to find this set, nor do they apply to claw-free graphs.

\subsection{Our Results}
The main result of this paper is that {\sc Induced Graph Matching} is fixed-parameter tractable on claw-free graphs when parameterized by $k$ for fixed connected graphs $H$.
It is important to note that requiring~$H$ to be fixed is essential, since the problem becomes $\mathsf{W}[1]$-hard when parameterized by $k + |V(H)|$, even for line graphs and co-bipartite graphs~\cite{GolovachPL12,GolovachPL2012}.
In the special case that $H$ is a fixed complete graph, we also show that the problem admits a polynomial kernel.
In contrast, we prove that the problem becomes $\mathsf{W}[1]$-hard on $K_{1,4}$-free graphs when parameterized by $k$, even if $|V(H)| = 1$.
These results both complement and tighten the above-mentioned known hardness results on the problem.

We also consider {\sc Induced Graph Matching} on a large and important subclass of claw-free graphs, namely on proper circular-arc graphs.
We prove that, if $H$ is connected and $G$ has a representation in which no three arcs cover the entire circle (i.e.~it is a proper Helly circular-arc graph), then the problem is polynomial-time solvable. However, if $H$ is connected but $G$ has no representation in which no three arcs cover the entire circle, then the problem is $\mathsf{NP}$-complete. If $H$ is not connected, then we observe that {\sc Induced Graph Matching} is $\mathsf{NP}$-complete on graphs $G$ that are in a subclass of proper circular-arc graphs, and show that it is fixed-parameter tractable on general proper circular-arc graphs $G$ when parameterized by $k$ and the number of connected components of $H$.

\subsection{Outline of the Main Algorithm}
First, we solve the {\sc Induced Graph Matching} problem on a subclass of claw-free graphs, called fuzzy circular-arc graphs (a superclass of proper circular-arc graphs).
On fuzzy circular-arc graphs, the {\sc Induced Graph Matching} problem was not previously known to be polynomial-time solvable for fixed~$H$.
We show that in fact it is, and give evidence why this might be the best possible result we can hope for.

Second, we prove a new decomposition theorem for claw-free graphs.
This new decomposition theorem becomes necessary, because both the original Chudnovsky-Seymour decomposition and the decomposition by Hermelin et al.~\cite{HermelinMLW11} require that certain structures have been removed from the claw-free graph, by some problem-specific preprocessing.
Yet, for the {\sc Induced Graph Matching} problem, it seems hard to get rid of these structures (so-called \emph{twins} and \emph{proper W-joins}), motivating our new decomposition theorem which can handle these structures.
We further show that this decomposition can be found in polynomial time, that is, a ``decomposition structure'' is computed that models the various parts of the input claw-free graph and their interactions.
The new decomposition theorem that we develop has the advantage that it is simpler to state than the earlier theorems by Chudnovsky-Seymour and Hermelin et al.~\cite{HermelinMLW11}.
In particular, it decomposes the input claw-free graph into only two graph classes, one of them being the aforementioned fuzzy circular-arc graphs. We stress that the decomposition is not problem-specific, and thus may be of independent interest.

Third, employing our refined decomposition theorem, we solve {\sc Induced Graph Matching} on claw-free graphs by applying the color-coding technique~\cite{AlonYZ1995}.
To give some intuition behind this approach, we recall that to solve {\sc Induced Graph Matching} on line graphs $G$ and $H$, we need to find pairwise vertex-disjoint induced subgraphs $H_1,\ldots,H_k$ in the pre-image of $G$ such that $H_{i}$ is isomorphic to the pre-image of $H$.
In particular, we show that we can find these $H_{i}$ in our structural decomposition for $G$.
However, if $G$ and $H$ are claw-free, there is no notion of pre-image and therefore no immediate relation between~$H$ and the $H_i$ that we would want to find.
Instead, we show that it is sufficient to find, in a coordinated manner and using color-coding, isomorphic copies of $H$ itself.
Thus, in some sense, color-coding allows us to reduce {\sc Induced Graph Matching} to {\sc Graph Matching}.

To obtain a polynomial kernel when $H$ is a fixed complete graph of order $h$, we reduce the size of the decomposition structure to $O(h^{4}k^{2})$.
We then construct a kernel that mimics an easy algorithm for {\sc Induced Graph Matching} on such reduced decomposition structures.
This kernel actually reduces the problem to an equivalent instance of {\sc Weighted Independent Set} of size polynomial in $k$, which in turn can be reduced to an equivalent instance of {\sc Induced Graph Matching} on claw-free graphs.
This approach substantially simplifies the approach that was used for the polynomial kernel for \textsc{Dominating Set} on claw-free graphs~\cite{HermelinMLW11}.

\paragraph{Organization}
In Section~\ref{sec:prelim}, we collect definitions used in the paper.
Section~\ref{sec:circ} presents a polynomial-time algorithm for {\sc Induced Graph Matching} on fuzzy circular-arc graphs when $H$ is a fixed connected graph.
It also considers the complexity of {\sc Induced Graph Matching} on circular-arc graphs.
Then, in Section~\ref{sec:algorithmicdecompositionofclawfreegraphs}, we prove a new algorithmic decomposition theorem for claw-free graphs.
The fixed-parameter algorithm for {\sc Induced Graph Matching} on claw-free graphs is given in Section~\ref{sec:algo}.
Proofs of the parameterized intractability of {\sc Induced Graph Matching} on $K_{1,4}$-free graphs are presented in Section~\ref{sec:hardness}.
We conclude in Section~\ref{sec:discussion}.
The $\mathsf{NP}$-hardness of {\sc Induced Graph Matching} for $H = K_3$ is proved in Appendix~\ref{sec:k3}.

\section{Preliminaries}
\label{sec:prelim}
All graphs considered in this paper are finite and undirected, but sometimes we consider multi-graphs with possible self-loops (we will explicitly mention when we consider such graphs).
For a (multi-)graph $G$, let $V(G)$ denote its vertex set and $E(G)$ its (multi-)set of edges.
Let $G$ be a graph.
We call $I \subseteq V(G)$ an \emph{independent set} if no two vertices of~$I$ are adjacent, and use $\alpha(G)$ to denote the size of a maximum independent set of $G$.
The neighborhood of a vertex $v \in V(G)$ is denoted by $N(v) = \{u \in V(G) \mid \{u,v\} \in E(G)\}$, and the closed neighborhood of $v$ is $N[v] = N(v) \cup \{v\}$.
We extend this notation to sets of vertices $X \subseteq V(G)$ by $N(X):= (\bigcup_{v \in X} N(v)) \setminus X$ and $N[X]:= N(X)\cup X$.
Given $X \subseteq V(G)$, the subgraph induced by $X$ is $G[X] = (X, E(G) \cap (X \times X))$, and the graph $G-X$ is the subgraph induced by $V(G) \setminus X$.
Then $G$ is \emph{claw-free} if $\alpha(G[N(v)]) \leq 2$ for any $v \in V(G)$.

Two simple graphs $G$ and $H$ are \emph{isomorphic} if there exists a bijection $\phi : V(H) \rightarrow V(G)$ such that $u$ and $v$ are adjacent in $H$ if and only if $\phi(u)$ and $\phi(v)$ are adjacent in $G$. We call an induced subgraph of $G$ that is isomorphic to $H$ an \emph{occurrence} of $H$ in $G$. Note that for every occurrence of $H$ in $G$ there is an injection $\phi : V(H) \rightarrow V(G)$ such that $u$ and $v$ are adjacent in $H$ if and only if $\phi(u)$ and $\phi(v)$ are adjacent in $G$.

We now define intersection graphs in general, and proper circular-arc graphs and several variants in particular.
For a family $\mathcal A$ of sets, a graph $G$ is the \emph{intersection graph} of $\mathcal A$ if each vertex corresponds to a set in $\mathcal A$ and there is an edge between two vertices if and only if the corresponding sets have non-empty intersection.
We call $\mathcal A$ a \emph{representation} of $G$.
We can then define (proper) circular-arc graphs. Consider a set of arcs on a circle (throughout, we assume that arcs are not a single point).
Then the set is called \emph{proper} if no arc is a subset of another. A graph is a \emph{(proper) circular-arc graph} if it is the intersection graph of a set of (proper) arcs on a circle. We call a set of arcs \emph{long} if the union of any three arcs does not cover the entire circle. The intersection graph of a long, proper set of arcs is a \emph{long proper circular-arc graph}.
The intersection graph of a set of (proper) intervals of a line is called a \emph{(proper) interval graph} (a set of intervals is proper if no interval is a subset of another). Note that the intersection graph of a set of (proper) arcs that does not cover the entire circle is a (proper) interval graph.

To define fuzzy circular-arc graphs, we start with the definition of a fuzzy intersection graph. Given a family $\mathcal{A}$ of sets, a graph $G$ is a \emph{fuzzy intersection graph} of $\mathcal A$ if there is an edge in $G$ between two vertices if the corresponding sets intersect in more than one element, and there is no edge if the corresponding sets do not intersect.
The fuzziness stems from what happens if the two sets intersect in precisely one element: the graph may have an edge or not.
We again call $\mathcal A$ a \emph{representation} of $G$.
Now, as before, consider a set of arcs of a circle. The set is \emph{almost proper} if any two arcs either have the same endpoints and cover the same part of the circle, or none of the two is a subset of the other. The set is \emph{almost strict} if, for any maximal set of arcs with the same endpoints, at most one of these endpoints is also an endpoint of an arc outside the set. Then a graph is a \emph{fuzzy circular-arc graph} if it is the fuzzy intersection graph of an almost-proper, almost-strict set of arcs on a circle. 
We sometimes also consider strict sets of arcs: a set of arcs is \emph{strict} if no two arcs share an endpoint.

Note that any intersection graph of an almost-proper set of arcs on a circle is a proper circular-arc graph. Hence, a fuzzy circular-arc graph without any fuzziness is just a proper circular-arc graph.
It is known that any fuzzy circular-arc graph is claw-free. Moreover, it follows from a result of Oriolo et al.~\cite{OrioloPS2012} that fuzzy circular-arc graphs can be recognized in polynomial time (in Theorem~\ref{thm:circ:fuzzy-reg}, we prove how this follows).

Finally, we use $\mathbb{N}$ to denote the set of nonnegative integers, i.e.~$\mathbb{N} = \{0,1,\ldots\}$.

\section{\textsc{Induced Graph Matching} on Proper and Fuzzy Circular-Arc Graphs}
\label{sec:circ}
In this section, we consider the complexity of {\sc Induced Graph Matching} on proper circular-arc graphs and on fuzzy circular-arc graphs.
We first completely analyze the complexity of this problem if $G$ is a proper circular-arc graph.
Then we observe that the problem is $\mathsf{W}[1]$-hard on fuzzy circular-arc graphs when parameterized by $|V(H)|$.
Finally, we give a polynomial-time algorithm on fuzzy circular-arc graphs when $H$ is a fixed connected graph.

\subsection{Proper Circular-Arc Graphs}
We consider the complexity of {\sc Induced Graph Matching} on proper circular-arc graphs.
These results generalize results by Heggernes et al.~\cite{HeggernesHMV2012} for {\sc Induced Subgraph Isomorphism} on proper interval graphs.
Recall that {\sc Induced Subgraph Isomorphism} asks, given two graphs $G$ and~$H$, whether $G$ contains an induced subgraph that is isomorphic to $H$. 
Also recall that a \emph{(proper) circular-arc graph} is the intersection graph of a set of (proper) arcs on a circle. If the union of any three arcs does not contain the entire circle, then the graph is called \emph{long}. If the union of all arcs does not cover the circle, then the graph is a \emph{(proper) interval graph}. 

\begin{theorem}
\label{thm:circ:int-int-poly}
{\sc Induced Graph Matching} on a proper interval graph $G$ and a connected proper interval graph $H$ can be solved in polynomial time.
\end{theorem}
\begin{proof}
Fix a representation of $G$. Throughout, we use vertices of $G$ and their corresponding intervals in this representation interchangeably.
The algorithm will follow an approach suggested by Cameron and Hell~\cite{CameronHell2006}. They show that if we construct for each occurrence of $H$ in $G$ a single interval that corresponds to the union of the intervals of the occurrence, then in the auxiliary graph consisting of all these single intervals it suffices to solve {\sc Independent Set} in order to obtain a solution for {\sc Induced Graph Matching} on the original graph. Therefore, we first show how to find all (relevant) occurrences of $H$ in $G$, and then reduce to an instance of {\sc Independent Set} on an auxiliary interval graph.

In the first step of the algorithm, we describe which occurrences of $H$ are relevant, and how to find all relevant occurrences in polynomial time.
Observe that any occurrence of $H$ in $G$ can be characterized by its leftmost and its rightmost interval in $G$. In particular, given two intervals $l, r$, all occurrences of $H$ in $G$ that have $l$ and $r$ as their leftmost and rightmost interval respectively can be considered equivalent for the purpose of {\sc Induced Graph Matching}. In particular, in any solution, we may freely exchange occurrences of $H$ from the same equivalence class without changing the feasibility of the solution. Therefore, it suffices to consider occurrences of $H$ in $G$ up to this equivalence relation. To enumerate all such occurrences of $H$ in $G$, we enumerate all pairs of vertices $l,r \in V(G)$ where the left endpoint of $l$ is to the left of the left endpoint of $r$. Then, remove from $G$ all vertices that have their left endpoint to the left of $l_{G}$ and all vertices that have their right endpoint to the right of $r_{G}$. Call the resulting graph $G'$, and observe that $G'$ is still a proper interval graph. Therefore, we can check in polynomial time using a result of Heggernes et al.~\cite{HeggernesHMV2012} whether there is an occurrence of $H$ in $G'$, and thus check whether there is an occurrence of $H$ in $G$ whose leftmost interval does not lie to the left of $l$ and whose rightmost interval does not lie to the right of $r$. In other words, we can indeed enumerate all possible occurrences of $H$ in $G$, up to the equivalence relation described above. Let $\mathcal{G}_\mathcal{H}$ denote this set of occurrences of $H$ in $G$.

We now follow the approach suggested by Cameron and Hell~\cite{CameronHell2006}. We construct an interval graph $G^*$ as follows: For each $G_H \in \mathcal{G}_\mathcal{H}$, we construct an interval from the position of the leftmost left endpoint in $V(G_H)$ to the rightmost right endpoint in $V(G_H)$. Note that since $H$ is connected, this interval corresponds to the union of the intervals of $G_H$. We let $G^*$ be the intersection graph of all these intervals. Note that since $|\mathcal{G}_\mathcal{H}| = O(|V(G)|^{2})$, the interval graph $G^*$ has polynomial size. Moreover, it is not difficult to see that any independent set in $G^*$ corresponds to an induced $H$-matching in $G$, and vice versa. Thus, as {\sc Independent Set} on interval graphs can be solved in polynomial time, we obtain the polynomial-time algorithm claimed in the theorem statement. \qed
\end{proof}
Possibly the running time of the above algorithm can be improved by merging the algorithm for {\sc Independent Set} on interval graphs with the given algorithm that finds relevant occurrences of $H$ in $G$.

Note that the requirement in Theorem~\ref{thm:circ:int-int-poly} above that $H$ is connected is crucial, since without this requirement already {\sc Induced Subgraph Isomorphism}, which is a special case of \textsc{Induced Graph Matching}, is $\mathsf{NP}$-complete~\cite{HeggernesHMV2012}. We extend Theorem~\ref{thm:circ:int-int-poly} to the case where both $G$ and $H$ are long proper circular-arc graphs, and $H$ is also required to be connected. First, we show that the restriction to long graphs is necessary. Recall that a graph $G$ is \emph{co-bipartite} if it is the complement of a bipartite graph.

\begin{theorem}
\label{thm:circ:np-complete}%
The {\sc Induced Subgraph Isomorphism} problem on connected proper circular-arc graphs is $\mathsf{NP}$-complete, even when the graphs are co-bipartite and have a representation in which there exist three arcs that cover the entire circle.
\end{theorem}

\begin{proof}
Damaschke~\cite{Damaschke1990} (see also the book of Garey and Johnson~\cite{GareyJ1979}) observed that {\sc Induced Subgraph Isomorphism} is $\mathsf{NP}$-complete even when both input graphs are a disjoint union of paths. The intuition behind that reduction is to reduce from {\sc $3$-Partition}, where one creates a path of length $a_{i}$ for each integer $a_{1},\ldots,a_{m} \geq 1$ of the instance ($H$ is the union of these paths), and creates $G$ as the union of $m/3$ paths of length $2+3(\sum_{i=1}^{m} a_{i})/m$ that capture the $m/3$ parts of the $3$-partition. We may assume that at least one integer $a_{i}$ is at least seven. Otherwise, there are a constant number of ways to make parts of the $3$-partition, and we can use exhaustive enumeration to solve the problem optimally in polynomial time.

Suppose that we are given an instance $(G,H)$ of {\sc Induced Subgraph Isomorphism} of the form described above. Let $\bar{G}$ and $\bar{H}$ be the complement of $G$ and $H$ respectively. Since the problem is closed under taking complements, $(\bar{G}, \bar{H})$ forms an equivalent instance of {\sc Induced Subgraph Isomorphism}. It remains to prove that $\bar{G}$ and $\bar{H}$ are connected proper circular-arc graphs that are co-bipartite and have a representation in which there exist three arcs that cover the entire circle.

To see that $\bar{G}$ and $\bar{H}$ are co-bipartite, connected, proper circular-arc graphs, note that a disjoint union of paths is a bipartite permutation graph, i.e.~a bipartite graph that is an intersection graph of a permutation diagram. Moreover, the complement of a bipartite permutation graph is a co-bipartite proper circular-arc graph~\cite{Spinrad2003}. It is easy to see that $\bar{G}$ and $\bar{H}$ are connected.

To see that $\bar{G}$ and $\bar{H}$ are connected proper circular-arc graphs that have a representation in which there exist three arcs that cover the entire circle, we show that the complement of a path is a proper circular-arc graph, and that the complement of a disjoint union of paths is a connected proper circular-arc graph. For the first property, let $P = \{v_{1},\ldots\}$ be a path. Then the complement of $P$ can be represented by arcs such that the arcs for $v_{2i+2}$ and $v_{2i+3}$ are copies of $v_{2i}$ and $v_{2i+1}$ that are slightly rotated. In particular, the arcs must be such that if $|E(P)| \geq 7$, then three arcs (say $v_{1},v_{4},v_{7}$) cover the entire circle. The second property can be obtained by appropriately rotating the representations of each path complement.
\qed
\end{proof}

Observe that a proper circular-arc graph that has a representation in which there exist two arcs that cover the entire circle is co-bipartite. Therefore, we can indeed complement Theorem~\ref{thm:circ:np-complete} by showing that {\sc Induced Subgraph Isomorphism} is polynomial-time solvable when $G$ and $H$ are long (i.e.~no three arcs cover the entire circle) and $H$ is connected.

\begin{theorem}
\label{thm:circ:circarc-cic-isi}
{\sc Induced Subgraph Isomorphism} on a long proper circular-arc graph $G$ and a connected proper circular-arc graph~$H$ can be solved in polynomial time.
\end{theorem}

\begin{proof}
We observe that long proper circular-arc graphs are actually proper Helly circular-arc graphs~\cite{McKee2003} (see also the paper by Lin et al.~\cite[Theorem~7]{LinSS2013}). Therefore, we can compute a long representation $\mathcal{A}$ of $G$ in linear time~\cite{LinSS2013}. Since the arcs in $\mathcal{A}$ satisfy the Helly property, we know that no maximal clique of $G$ covers the entire circle and that the arcs of any maximal clique have a non-empty intersection.

Now suppose that $H$ occurs in $G$. Observe that there exists a maximal clique $K_H$ of $H$ in any occurrence of $H$ in $G$ that is mapped to a subset of a maximal clique $K_G$ of $G$. As no clique of $G$ covers the entire circle, there is a point $p$ such that the set of all arcs containing $p$ is equal to $K_{G}$ and is included in all arcs of $K_H$. Let $G' = G - K_G$ and $H' = H - K_H$. Then~$G'$ contains an occurrence of $H'$, and both $G'$ and $H'$ are proper interval graphs.

This suggests the following algorithm. Compute a long representation of $G$ and $H$ in polynomial time~\cite{LinSS2013} (if such a representation does not exist for $H$, we can answer ``no'' immediately). Consider two points in the representation of $G$ or $H$ to be equivalent if the set of arcs containing the one point is the same as the set of arcs containing the other point. Clearly, this equivalence relation has $O(|V(G)|)$ and $O(|V(H)|)$ equivalence classes on $G$ and $H$ respectively. For any point $p_{G}$ in an equivalence class of $G$, let $G_{p_{G}}$ denote the subgraph of $G$ obtained by removing an infinitesimally small part of each arc around $p_{G}$ for each arc containing $p_{G}$. This makes $G_{p_{G}}$ a proper interval graph, while effectively duplicating the set $A_{p_{G}}$ of all arcs containing $p_{G}$. Now construct two new cliques of size $1+\max\{|V(G)|,|V(H)|\}$ each, and make each adjacent to a copy of $A_{p_{G}}$, and call the resulting graph~$G_{p_{G}}'$. In a similar manner, we can consider a point $p_{H}$ and construct a graph $H_{p_{H}}'$. Note that both $G_{p_{G}}'$ and $H_{p_{H}}'$ are connected proper interval graphs. Hence, we can determine in polynomial time whether $H_{p_{H}}'$ occurs in~$G_{p_{G}}'$ using the algorithm by Heggernes et al.~\cite{HeggernesHMV2012}. Moreover, by the observations of the previous paragraph, $H_{p_{H}}'$ occurs in~$G_{p_{G}}'$ for some choice of $p_{G}$ and $p_{H}$ if and only if $H$ occurs in $G$. The theorem follows. \qed
\end{proof}

\begin{theorem}
\label{thm:circ:circarc-circarc-poly}
{\sc Induced Graph Matching} on a long proper circular-arc graph $G$ and a connected proper circular-arc graph~$H$ can be solved in polynomial time.
\end{theorem}
\begin{proof}
Compute a representation of $G$, using that long proper circular-arc graphs are actually proper Helly circular-arc graphs~\cite{LinSS2013}. If $G$ contains at least two independent occurrences of $H$, then there is a point $p$ on the circle such that no occurrence of $H$ contains an arc that contains $p$. Again, define an equivalence relation where two points are equivalent if the set of arcs containing the one point is the same as the set of arcs containing the other point. The number of equivalence classes is $O(|V(G)|)$. For each equivalence class, consider a point $p$, and consider the graph $G_{p}$ obtained by removing all arcs from $G$ that contain $p$. As $G_{p}$ is a proper interval graph, we can apply Theorem~\ref{thm:circ:int-int-poly} to $G_{p}$ and $H$, and return the largest solution returned (if any) over all equivalence classes. If no solution is returned, then $G$ contains at most one occurrence of $H$. This can be checked by computing whether $G$ has an induced subgraph isomorphic to $H$ using Theorem~\ref{thm:circ:circarc-cic-isi}. \qed
\end{proof}

Finally, we show that we can obtain a parameterized result if $H$ is not connected.

\begin{theorem}
\label{thm:circ:circarc-circarc-fpt}
  {\sc Induced Graph Matching} on a proper circular-arc graph $G$ and a disconnected proper circular-arc graph~$H$ is fixed-parameter tractable when parameterized by $k$ and the number of connected components of $H$.
\end{theorem}
\begin{proof}
Since~$H$ is not connected, in any representation of $H$ there must be a point of the circle that is not covered by an arc. If we cut the circle open on this point, then we obtain a representation of $H$ as intervals of a line. Therefore, $H$ is a proper interval graph.

Consider any representation of $G$. As $H$ is not connected, there must be a point such that no occurrence of $H$ in a maximum induced $H$-matching uses an arc containing this point.
  Up to equivalence, there are $O(|V(G)|)$ such points that we need to consider.
  For any such point $p$, let $G_{p}$ be the graph obtained from $G$ by removing all arcs containing $p$.
  Note that~$G_{p}$ is a proper interval graph.
  Then create $H'$ as the disjoint union of $k$ copies of $H$, and find $H'$ as an induced subgraph in $G_{p}$.
  Heggernes et al.~\cite{HeggernesHMV2012} have shown that finding a proper interval graph $H'$ as an induced subgraph in another proper interval graph $G_{p}$ is fixed-parameter tractable when parameterized by the number of connected components of $H'$.
  Since the number of connected components of $H'$ is equal to~$k$ times the number of connected components of $H$, and there is a $G_{p}$ with an induced $H$-matching of size~$k$ if and only if there is an induced $H$-matching of size $k$ in $G$, the theorem follows.
\qed
\end{proof}

Note that {\sc Induced Graph Matching} is $\mathsf{NP}$-hard if both $G$ and $H$ are proper interval graphs, even if $k=1$~\cite{Damaschke1990}.
Hence, parameterization by some property of $H$ is really necessary in Theorem~\ref{thm:circ:circarc-circarc-fpt}.
We discuss further possibilities for improving on the results in this section in Section~\ref{sec:discussion}.

\subsection{Fuzzy Circular-Arc Graphs}
The generalization from proper circular-arc graphs to fuzzy circular-arc graphs makes {\sc Induced Graph Matching} substantially harder.
This is due to co-bipartite graphs.

\begin{theorem}
\label{thm:circ:co-hard}
  {\sc Induced Graph Matching} is $\mathsf{NP}$-hard, and $\mathsf{W}[1]$-hard when parameterized by $|V(H)|$, even if $k=1$ and both $G$ and $H$ are connected co-bipartite graphs, or both $G$ and $H$ are connected fuzzy circular-arc graphs, or both $G$ and $H$ are connected claw-free graphs.
\end{theorem}
\begin{proof}
Golovach et al.~\cite{GolovachPL2012} implicitly prove that {\sc Induced Subgraph Isomorphism} is $\mathsf{NP}$-hard, and $\mathsf{W}[1]$-hard when parameterized by $|V(H)|$, on co-bipartite graphs. As {\sc Induced Graph Matching} for parameter $k = 1$ is the {\sc Induced Subgraph Isomorphism} problem, the result for co-bipartite graphs follows. It then suffices to observe that co-bipartite graphs are both fuzzy circular-arc graphs\footnote{We can explicitly prove this using the language of Section~\ref{sec:algorithmicdecompositionofclawfreegraphs}. Note that any co-bipartite graph is a thickening of a single semi-edge, a single vertex, or two independent vertices. Moreover, the graph consisting of a single semi-edge, a single vertex, or two independent vertices is a circular interval trigraph. Hence, using Lemma~\ref{lem:circ:circ-model}, we can see that any co-bipartite graph is a fuzzy circular-arc graph.} and claw-free graphs.
\qed\end{proof}
Therefore, we consider the case when $H$ is fixed.

\begin{theorem}
\label{thm:algo:circ}
  {\sc Induced Graph Matching} on fuzzy circular-arc graphs $G$ can be solved in polynomial time when $H$ is a fixed, connected graph.
\end{theorem}
\begin{proof}
  We generalize the approach of Theorem~\ref{thm:circ:int-int-poly}.
  First, observe that since $H$ is fixed, we can find the set~$\mathcal{H}$ of all occurrences of $H$ in $G$ in polynomial time.
  We call two occurrences of $H$ \emph{compatible} if no vertex of one occurrence is equal to or a neighbor of a vertex of the other.

  We may assume that $H$ occurs at least once in $G$.
  Pick an arbitrary $H^{*} \in \mathcal{H}$, and remove $N[H^{*}]$ from $G$.
  Since fuzzy circular-arc graphs are closed under vertex deletion, the resulting graph $G' = G - N[H^*]$ is also a fuzzy circular-arc graph.

  Now find a representation of $G'$ using Theorem~\ref{thm:circ:fuzzy-reg} (which essentially rephrases a result by Oriolo et al.~\cite{OrioloPS2012}).
  For each $H' \in \mathcal{H}\setminus\{H^*\}$, let $l(H')$ (resp.\ $r(H')$) be the leftmost (resp.\ rightmost) endpoint of any arc of $H'$.
  Let $\mathcal{P}'$ be the set of endpoints of the arcs of $\mathcal{A}'$.
  For each point $p' \in \mathcal{P}'$, let $\mathcal{H}_{p'}$ be the set of $H' \in \mathcal{H}\setminus\{H^*\}$ for which $r(H') = p'$.
  We use $\mathcal{H}_{p'}^{1}, \ldots,\mathcal{H}_{p'}^{|\mathcal{H}_{p'}|}$ to denote the elements of $\mathcal{H}_{p'}$.
  Finally, since we removed one occurrence $H^*$ from $G$ to obtain $G'$, we can uniquely order the points in $\mathcal{P}'$ from left to right as $p_{1},\ldots,p_{|\mathcal{P}'|}$.

We describe a dynamic programming approach that solves {\sc Induced Graph Matching}.
For each $i = 1,\ldots,|\mathcal{P}'|$ and for each $j = 1,\ldots,|\mathcal{H}_{p_{i}}|$, we will compute $M[i,j]$ as the size of a largest induced $H$-matching in the subgraph of $G'$ induced by the vertices whose arcs have their right endpoint on or to the left of $p_{i}$, such that this induced $H$-matching contains $\mathcal{H}_{p_{i}}^{j}$.
For simplicity, we set $M[0,1] = 0$ and let $\mathcal{H}_{p_{0}}^{1}$ be a ``fake'' occurrence that is compatible with all occurrences of $H$.
We can then compute $M$ using the following formula:
  \begin{equation*}
  M[i,j] = 
              \displaystyle 1 + \max_{0 \leq i^{*} < i}\ \max_{\stackrel{j^{*} = 1,\ldots,|\mathcal{H}_{p_{i^{*}}}|}{\mathcal{H}_{p_{i}}^{j}, \mathcal{H}_{p_{i^{*}}}^{j^{*}} \mbox{compatible} }}\ M[i^{*}, j^{*}]
  \end{equation*}
Then the size of a largest induced $H$-matching of $G-N[H^{*}]$ follows by computing:
  \begin{equation*}
  \max_{\stackrel{i = 1,\ldots,|\mathcal P'|}{j =1,\ldots,|\mathcal{H}_{p_{i}}|}} \{M[i,j]\} \enspace.
  \end{equation*}
  We can adapt this algorithm to obtain the induced $H$-matching that achieves the maximum.
  By running this dynamic programming algorithm for all possible choices of $H^{*}$, we obtain a polynomial-time algorithm (when $H$ is fixed) to solve {\sc Induced Graph Matching}.
\qed
\end{proof}

\section{Algorithmic Decomposition of Claw-Free Graphs}
\label{sec:algorithmicdecompositionofclawfreegraphs}
In this section we prove an algorithmic decomposition theorem of claw-free graphs.
We start by introducing some notation and stating the decomposition theorem, and then go on to prove the theorem.

The backbone for the decomposition theorem is formed by so-called strips and strip-structures, concepts first introduced by Chudnovsky and Seymour~\cite{ChudnovskyS2007,ChudnovskyS2008a,ChudnovskyS2008b,ChudnovskyS2008c,ChudnovskyS2008d,ChudnovskyS2010,ChudnovskyS2012}. We start by defining these notions, in a manner that is specialized towards their use for claw-free graphs. Throughout, we ask the reader to refer to Fig.~\ref{fig:stripstructure} for accompanying illustrations.

\begin{definition}
A \emph{strip} $(J,Z)$ of a graph $G$ is a tuple with a graph $J$ and an independent set $Z \subseteq V(J)$, such that $\emptyset \neq (V(J) \setminus Z) \subseteq V(G)$, the subgraphs of $J$ and $G$ induced by $V(J) \setminus Z$ are isomorphic, and for each $z \in Z$, $N(z)$ is a nonempty clique in $J$.
\end{definition}
It is crucial to observe that the vertices of $Z$ are \emph{not} vertices of $G$, but auxiliary vertices. Moreover, the set $Z$ may be empty; however, in the way that strips are used later on, this occurs only in one situation.

We introduce several terms for the different parts of a strip $(J,Z)$. The set $N(z)$ for a $z \in Z$ is called a \emph{boundary}, while the (possibly empty) graph $J - N[Z]$ is called the \emph{interior} of $(J,Z)$.

We also introduce two special types of strips.
First, we call a strip $(J,Z)$ a \emph{spot} if $J$ is a three-vertex path and $Z$ consists of its ends.
Second, we call a strip $(J,Z)$ a \emph{stripe} if no vertex of $J$ is adjacent to more than one $z \in Z$. This implies that for stripes the sets $N(z)$ for all $z \in Z$ are pairwise disjoint.

Before we can define strip-structures, we need the notion of a hypergraph.
A \emph{hypergraph} $\mathcal R$ consists of a set of vertices $V(\mathcal R)$ and a multi-set of hyperedges $E(\mathcal R)$ such that each $e \in E(\mathcal R)$ is a subset of $V(\mathcal R)$.
We explicitly allow empty hyperedges.
Note that if $|e| \leq 2$ for every $e \in E(\mathcal R)$, then $\mathcal R$ can be considered as a multi-graph, possibly with self-loops.

\begin{definition} \label{def:strip-structure}
A \emph{strip-structure} of a claw-free graph $G$ consists of a hypergraph~$\mathcal R$ (called the \emph{strip-graph}) with at least one edge, and for each hyperedge $e \in E(\mathcal R)$ a strip $(J_{e},Z_{e})$ of $G$ such that
\begin{itemize}
\item the sets $V(J_{e}) \setminus Z_{e}$ over all $e \in E(\mathcal R)$ partition $V(G)$;
\item for each $e \in E(\mathcal R)$, $J_{e}$ is a claw-free graph;
\item for each $e \in E(\mathcal R)$ and for each $r \in e$, there is a unique $z_{e}^{r} \in Z_{e}$, and $Z_{e} = \{ z_{e}^{r} \mid r \in e \}$;
\item for each $r \in V(\mathcal R)$, the union of $N(z_{e}^{r})$ over all $e \in E(\mathcal R)$ for which $r \in e$ induces a clique \emph{$C(r)$} in $G$;
\item if $u,v \in V(G)$ are adjacent, then $u,v \in V(J_{e})$ for some $e \in E(\mathcal R)$ or $u,v \in C(r)$ for some $r \in V(\mathcal R)$.
\end{itemize}
\end{definition}
For given $e \in E(\mathcal R)$ and $r \in e$, we call $N(z_{e}^{r})$ the \emph{boundary corresponding to $e$ and $r$}.
For given $r$, we call $C(r)$ the \emph{clique on $r$}.
Finally, to avoid confusion, we will from now on refer to hyperedges of a strip-graph as \emph{strip-edges} and vertices of a strip-graph as \emph{strip-vertices}; in this way, it is easy to distinguish whether we speak of vertices of the graph $G$ or of a strip-graph.

A crucial part of the definition of a strip-structure is the last item, which regulates how edges of $G$ are distributed over the strips.
Essentially, an edge can only go between a vertex of $J_{e} \setminus Z_{e}$ and a vertex of $J_{e'} \setminus Z_{e'}$ for strip-edges $e \not= e'$ of $\mathcal{R}$ if $e \cap e' \not= \emptyset$; even then these adjacent vertices must be part of the boundaries of $e$ and $e'$ that correspond to a strip-vertex of $e \cap e'$. In fact, the fourth item of the definition implies that the union of the boundaries corresponding to each strip-vertex of $e \cap e'$ is a clique. All edges of $G$ not in a clique on $r$ for some $r \in V(\mathcal{R})$ must be contained in $J_{e}[V(J_{e}) \setminus Z_{e}]$ for some $e \in E(\mathcal{R})$. Here, it is important to recall that the definition of a strip implies that $J_{e}[V(J_{e}) \setminus Z_{e}]$ and $G[V(J_{e}) \setminus Z_{e}]$ are isomorphic.
Some notions around strip-structures are illustrated in Fig.~\ref{fig:stripstructure}.

We note that any claw-free graph $G$ always has a trivial strip-structure consisting of a strip-graph with a single strip-edge $e$ and no strip-vertices, and of a strip $(J_{e},Z_{e}) = (G,\emptyset)$. The main result of this section, however, gives a much stronger strip-structure for claw-free graphs:

\tikzset{vertex/.style={minimum size=2mm,circle,fill=black,draw,inner sep=0pt},
        decoration={markings,mark=at position .5 with
{\arrow[black,thick]{stealth}}}}

\begin{figure}[t]
 \centering
 \begin{tikzpicture}[scale=0.5]
   \draw[fill=gray,fill opacity=0.2] (0,-4.5) ellipse (3.2cm and 1.4cm);
   \draw[fill=gray] (2,-4.5) ellipse (0.5cm and 1cm);
   \draw[fill=gray] (-2,-4.5) ellipse (0.5cm and 1cm);
   \draw[fill=gray, fill opacity=0.2] (-3,-2) ellipse (0.5cm and 0.5cm);
   \draw[fill=gray] (-3,-2) ellipse (0.35cm and 0.35cm);
   \draw[fill=gray, fill opacity=0.2] (-1,-2) ellipse (0.5cm and 0.5cm);
   \draw[fill=gray] (-1,-2) ellipse (0.35cm and 0.35cm);
   \draw[fill=gray, fill opacity=0.2] (1,-2) ellipse (0.5cm and 0.5cm);
   \draw[fill=gray] (1,-2) ellipse (0.35cm and 0.35cm);
   \draw[fill=gray, fill opacity=0.2] (3,-2) ellipse (0.5cm and 0.5cm);
   \draw[fill=gray] (3,-2) ellipse (0.35cm and 0.35cm);
   \draw[fill=gray, fill opacity=0.2] (2.5,-1) ellipse (0.5cm and 0.5cm);
   \draw[fill=gray] (2.5,-1) ellipse (0.35cm and 0.35cm);
   \draw[fill=gray, fill opacity=0.2] (-3,0) ellipse (0.5cm and 0.5cm);
   \draw[fill=gray] (-3,0) ellipse (0.35cm and 0.35cm);
            \draw[fill=gray,fill opacity=0.2] (-0.5,1.5) ellipse (1.5cm and 1.5cm);
            \draw[fill=gray] (-1,1) ellipse (0.45cm and 0.45cm);
            \draw[fill=gray] (0,1) ellipse (0.45cm and 0.45cm);
  \node (a) at (-3.75,-2){$a$};
  \node (b) at (-3.75,0){$b$};
  \node (c) at (-2.3,1.75){$c$};
  \node (d) at (-1,-2.75){$d$};
  \node (e) at (1,-2.75){$e$};
  \node (f) at (3.25,-1){$f$};
  \node (g) at (3.75,-2){$g$};
  \node (h) at (0,-6.5){$h$};
   \node (6) at (-2,-4) [vertex]{};
   \node (7) at (2,-4) [vertex]{};
   \node (8) at (-2,-5) [vertex]{};
   \node (9) at (2,-5) [vertex]{};
   \draw[thick] (7)--(9);
   \draw[thick] (6)--(8);
   \node (10) at (-1,-4) [vertex]{};
   \node (11) at (-0.5,-5) [vertex]{};
   \node (12) at (0.5,-4) [vertex]{};
   \draw[thick] (6)--(10);
   \draw[thick] (6)--(11);
   \draw[thick] (8)--(11);
   \draw[thick] (11)--(12);
   \draw[thick] (10)--(11);
   \draw[thick] (10)--(12);
   \draw[thick] (12)--(7);
   \draw[thick] (12)--(9);
   \node (13) at (3,-2) [vertex]{};
   \node (14) at (-3,-2) [vertex]{};
   \draw[thick] (13)--(7);
   \draw[thick] (13)--(9);
   \draw[thick] (14)--(6);
   \draw[thick] (14)--(8);
   \node (15a) at (-1,-2) [vertex]{};
   \draw[thick] (15a)--(13);
   \draw[thick] (15a)--(14);
   \node (15b) at (1,-2) [vertex]{};
   \node (16) at (-1,1) [vertex]{};
   \node (17) at (0,1) [vertex]{};
   \draw[thick] (16)--(17);
   \node (18) at (2.5,-1) [vertex]{};
   \draw[thick] (18)--(13);
   \draw[thick] (18)--(15b);
   \draw[thick] (16)--(15a);
   \draw[thick] (16)--(14);
   \draw[thick] (18)--(17);
   \node (19) at (-3,0) [vertex]{};
   \draw[thick] (19)--(14);
   \draw[thick] (19)--(16);
   \draw[thick] (13)--(14);
   \node (20) at (-0.5,2) [vertex]{};
   \draw[thick] (20)--(17);
   \draw[thick] (20)--(16);
   \draw[thick] (15a)--(19);
 \end{tikzpicture}
 \tikzset{every loop/.style={min distance=20mm,looseness=10}}
 \begin{tikzpicture}
  \node (0) at (-3,1) [vertex]{};
  \node (1) at (0,1) [vertex]{};
  \node (2) at (-3,3) [vertex]{};
  \node (3) at (0,3) [vertex]{};
  \node (4) at (-1.5,3) [vertex]{};
  \draw[thick] (0)--(1);
  \draw[thick] (1)--(3);
  \draw[thick] (2)--(0);
  \draw[thick] (2)--(3);
  \node (5) at (-1.5,4) [vertex]{};
  \draw[thick] (2)--(5);
  \draw[thick] (3)--(5);
  \draw[thick,in=135, out=45, loop,rotate=45] (2.center) to ();
  \node (a) at (-3.25,2){$a$};
  \node (b) at (-3,5,3){$b$};
  \node (c) at (-2.25,3.75){$c$};
  \node (d) at (-2.25,2.75){$d$};
  \node (e) at (-0.75,2.75){$e$};
  \node (f) at (-0.5,3.75){$f$};
  \node (g) at (0.25,2){$g$};
  \node (h) at (-1.5,0.75){$h$};
 \end{tikzpicture}
 \qquad
 \begin{tikzpicture}[scale=0.5]
   \draw[fill=gray, fill opacity=0.2] (0,0) ellipse (1.5cm and 1cm);
   \draw[fill=gray] (0,0) ellipse (0.5cm and 0.5cm);
   \node (1) at (0,0) [vertex] {};
   \node (2) at (-2,0) [vertex] {};
   \node (3) at (2,0) [vertex] {};
   \draw[thick] (1)--(2);
   \draw[thick] (1)--(3);
   \node at (0,-1.5){spot};
   \node at (-2,-0.8){$Z$};
   \node at (2,-0.8){$Z$};
   \draw[fill=gray, fill opacity=0.2] (0,-4.5) ellipse (2.5cm and 2cm);
   \draw[fill=gray] (2,-4.5) ellipse (0.5cm and 1cm);
   \draw[fill=gray] (-2,-4.5) ellipse (0.5cm and 1cm);
   \node (4) at (-3,-4.5) [vertex]{};
   \node (5) at (3,-4.5) [vertex]{};
   \node (6) at (-2,-4) [vertex]{};
   \node (7) at (2,-4) [vertex]{};
   \node (8) at (-2,-5) [vertex]{};
   \node (9) at (2,-5) [vertex]{};
   \draw[thick] (7)--(9);
   \draw[thick] (6)--(8);
   \draw[thick] (4)--(6);
   \draw[thick] (4)--(8);
   \draw[thick] (5)--(7);
   \draw[thick] (5)--(9);
   \node (10) at (-1,-4) [vertex]{};
   \node (11) at (-0.5,-5) [vertex]{};
   \node (12) at (0.5,-4) [vertex]{};
   \draw[thick] (6)--(10);
   \draw[thick] (6)--(11);
   \draw[thick] (8)--(11);
   \draw[thick] (11)--(12);
   \draw[thick] (10)--(11);
   \draw[thick] (10)--(12);
   \draw[thick] (12)--(7);
   \draw[thick] (12)--(9);
   \node at (0,-7){stripe};
   \node at (-3.2,-5.3){$Z$};
   \node at (3.0,-5.3){$Z$};
 \end{tikzpicture}
 \quad
 \quad
\caption{Illustrations for strip-structures and the two types of strips. In the left illustration, a claw-free graph is pictured. We have marked a strip-structure (for each strip, we left out vertices of $Z$): each light-gray ellipse corresponds to a strip, and each dark-gray ellipse corresponds to a boundary.
In the middle illustration, we pictured the strip-graph that corresponds to the strip-structure pictured in the left illustration. We remark that strip-graphs are \emph{not} necessarily claw-free, as the illustration shows.
In the right illustration, the two types of strips are pictured: spots and stripes. Again, the light-gray ellipses mark the set of vertices that are part of $G$, i.e.~the vertices of $Z$ do not belong to $G$, and the dark-gray ellipses mark the boundaries of the strips. We note that a spot always looks as pictured, and that the single vertex in the middle of the path is the boundary of the spot. Strips $a$, $d$, $e$, $f$, and $g$ of the left illustration correspond to spots. In the example of a stripe (note that it corresponds to the part $h$ of the left illustration), each boundary (and its associated vertex of $Z$) forms a clique. Strips $b$ and $c$ in the left illustration are also stripes, with one and two boundaries respectively.}
\label{fig:stripstructure}
\end{figure}
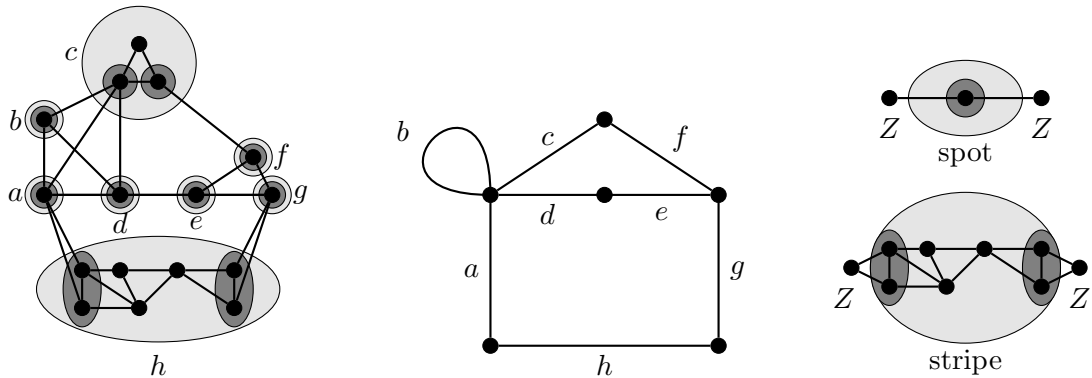

\begin{theorem}
\label{thm:prelim:main}%
  Let $G$ be a connected claw-free graph that is not a fuzzy circular-arc graph and has $\alpha(G) > 4$.
  Then $G$ admits a strip-structure such that each strip is either a spot, or a stripe $(J,Z)$ with $1 \leq |Z| \leq 2$ that is a fuzzy circular-arc graph or satisfies $\alpha(J) \leq 4$.
  Moreover, we can find such a strip-structure in polynomial time.
\end{theorem}

In order to prove this theorem, Subsection~\ref{sec:trigraphs:def} describes a convenient tool to work with claw-free graphs, the so-called trigraphs. In Subsection~\ref{sec:trigraphs:special}, we then define several special trigraphs and present some of their properties. Subsection~\ref{sec:trigraphs:joins} gives the structures through which we decompose claw-free graphs. Then, Subsection~\ref{sec:trigraphs:end} gives several auxiliary results, before giving the proof of Theorem~\ref{thm:prelim:main} in Subsection~\ref{sec:trigraphs:main}.

\subsection{Trigraphs} \label{sec:trigraphs:def}
The basic tool to describe the structure of claw-free graphs are trigraphs. We borrow this tool and its associated terminology from Chudnovsky and Seymour~\cite{ChudnovskyS2007,ChudnovskyS2008a,ChudnovskyS2008b,ChudnovskyS2008c,ChudnovskyS2008d,ChudnovskyS2010,ChudnovskyS2012}. The notion of a trigraph extends the standard notion of a graph.
A \emph{trigraph} consists of a set of vertices and a set of two types of edges, normal edges and \emph{semi-edges}, and each vertex can be incident on at most one semi-edge.
Observe that any graph is a trigraph, and that any trigraph without semi-edges is a graph.

We call two vertices of a trigraph \emph{strongly adjacent} if there is a normal edge between them, \emph{semi-adjacent} if there is a semi-edge between them, and \emph{strongly anti-adjacent} if there is no edge between them.
Two vertices are called \emph{adjacent} if they are strongly adjacent or semi-adjacent, and called \emph{anti-adjacent} if they are semi-adjacent or strongly anti-adjacent.

A trigraph $G$ is a \emph{thickening} of a trigraph $G'$ if there is a set $\mathcal{X} = \{X_{v'} \subseteq V(G) \mid v' \in V(G')\}$ such that each~$X_{v'}$ is nonempty, the $X_{v'}$ partition $V(G)$, and
\begin{itemize}
  \item if $u'$ is strongly adjacent to $v'$ in $G'$, then in $G$ each $u \in X_{u'}$ is strongly adjacent to each $v \in X_{v'}$;
  \item if $u'$ is strongly anti-adjacent to $v'$ in $G'$, then in $G$ each $u \in X_{u'}$ is strongly anti-adjacent to each $v \in X_{v'}$;
  \item if $u'$ is semi-adjacent to $v'$ in $G'$, then in $G$ there exist $u_{1},u_{2} \in X_{u'}$, $v_{1}, v_{2} \in X_{v'}$ such that $u_{1}$ is adjacent to $v_{1}$ and $u_{2}$ is anti-adjacent to $v_{2}$ (note that possibly $u_{1} = u_{2}$ or $v_{1}=v_{2}$);
  \item in $G$ any two vertices in $X_{v'}$ are strongly adjacent, for any $X_{v'}\in\mathcal X$.
\end{itemize}
We sometimes talk about the thickening $\mathcal{X}$ of $G'$ to $G$.

Let $G$ be a trigraph.
Given disjoint sets $A,B \subseteq V(G)$, we say that $A$ is \emph{(strongly) complete} to $B$ if each vertex of $A$ is (strongly) adjacent to each vertex of $B$.
Similarly, we define the notions of \emph{anti-complete} and \emph{strongly anti-complete}.
A set $C \subseteq V(G)$ is a \emph{(strong) clique} if every pair of vertices of~$C$ is (strongly) adjacent.
A set $I \subseteq V(G)$ is a \emph{(strong) independent set} if every pair of vertices of $C$ is (strongly) anti-adjacent.
We use $\alpha(G)$ to denote the size of a largest independent set of $G$.

For any $X \subseteq V(G)$, $G[X]$ is the subgraph of $G$ induced by $X$.
The notion of isomorphic trivially extends to trigraphs.
If $G[X]$ is isomorphic to some trigraph $G'$, then~$G'$ is said to be an \emph{induced subtrigraph} of $G$.

A \emph{claw} is the trigraph with four vertices $c,c_{1},c_{2},c_{3}$, where $c$ is complete to $c_{1}$, $c_{2}$, and $c_{3}$, and $\{c_{1},c_{2},c_{3}\}$ is an independent set.
Then~$G$ is \emph{claw-free} if no induced subtrigraph of $G$ is isomorphic to a claw.

We observe that if a trigraph has no semi-edges (and thus is a graph), then the notions of adjacent, anti-adjacent, complete, anti-complete, clique, independent set, claw, and claw-free behave exactly as expected from the usual definitions of these terms.

\subsection{Special Trigraphs} \label{sec:trigraphs:special}
Chudnovsky and Seymour~\cite{ChudnovskyS2007,ChudnovskyS2008a,ChudnovskyS2008b,ChudnovskyS2008c,ChudnovskyS2008d,ChudnovskyS2010,ChudnovskyS2012} identified eight special classes of trigraphs: $\mathcal{S}_{0},\ldots,\mathcal{S}_{7}$. We describe only those classes that are relevant to this paper, namely $\mathcal{S}_{0}$, $\mathcal{S}_{2}$, and $\mathcal{S}_{3}$. The interested reader is referred to the papers by Chudnovsky and Seymour~\cite{ChudnovskyS2007,ChudnovskyS2008a,ChudnovskyS2008b,ChudnovskyS2008c,ChudnovskyS2008d,ChudnovskyS2010,ChudnovskyS2012} for the definition of all other classes.

The class $\mathcal{S}_{0}$ consists of all line trigraphs.
A \emph{line trigraph} $G$ of a graph $G'$ has the set of edges of $G'$ as its vertex set. The edge set of $G$ is defined as follows. If two edges $e,f$ of $G'$ are incident on the same vertex and this vertex has degree two, then there is a normal edge or a semi-edge between $e,f$ in $G'$. If two edges $e,f$ of $G'$ are incident on the same vertex and this vertex has degree at least three, then there is a normal edge between $e,f$ in $G'$. If two edges $e,f$ are not incident on a common vertex, then there is no edge between $e,f$ in $G'$.

The class $\mathcal{S}_{2}$ consists all XX-trigraphs.
We say that $G$ is an \emph{XX-trigraph} if it can be obtained by removing any subset $X$ of $\{v_{7},v_{11},v_{12},v_{13}\}$ from the following trigraph on vertex set $v_{1},\ldots,v_{13}$:
\begin{itemize}
  \item $v_{i}$ is adjacent to $v_{i+1}$ for $i = 1,\ldots,5$ and $v_{6}$ is adjacent to $v_{1}$; also $v_{i}$ is anti-adjacent to $v_{j}$ for each $i = 1,\ldots,4$ and each $i+2 \leq j \leq 6$,
  \item $v_{7}$ is strongly adjacent to $v_{1}$ and $v_{2}$,
  \item $v_{8}$ is strongly adjacent to $v_{4}$, $v_{5}$, and possibly adjacent to $v_{7}$,
  \item $v_{9}$ is strongly adjacent to $v_{1}$, $v_{2}$, $v_{3}$, and $v_{6}$,
  \item $v_{10}$ is strongly adjacent to $v_{3}$, $v_{4}$, $v_{5}$, and $v_{6}$, and adjacent to $v_{9}$,
  \item $v_{11}$ is strongly adjacent to $v_{1}$, $v_{3}$, $v_{4}$, $v_{6}$, $v_{9}$, and $v_{10}$,
  \item $v_{12}$ is strongly adjacent to $v_{2}$, $v_{3}$, $v_{5}$, $v_{6}$, $v_{9}$, and $v_{10}$,
  \item $v_{13}$ is strongly adjacent to $v_{1}$, $v_{2}$, $v_{4}$, $v_{5}$, $v_{7}$, and $v_{8}$.
\end{itemize}
All other pairs of vertices are strongly anti-adjacent. An illustration of XX-trigraphs is given Fig.~\ref{fig:xxtrigraphs}. We need the following observation about XX-trigraphs.

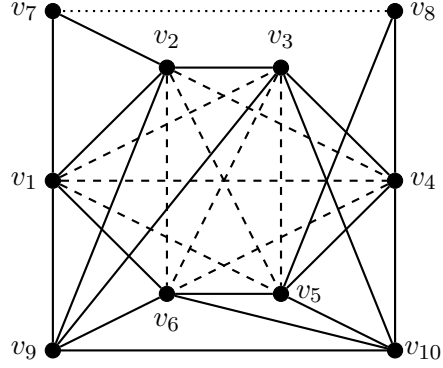
\begin{figure}[t]
  \begin{center}
    \begin{tikzpicture}[scale=1.5]
      \node (v1) at (1,1)[vertex]{};
      \node (v1l) at (0.75,1){$v_{1}$};
      \node (v2) at (2,2)[vertex]{};
      \node (v2l) at (2,2.25){$v_{2}$};
      \node (v3) at (3,2)[vertex]{};
      \node (v3l) at (3,2.25){$v_{3}$};
      \node (v4) at (4,1)[vertex]{};
      \node (v4l) at (4.25,1){$v_{4}$};
      \node (v5) at (3,0)[vertex]{};
      \node (v5l) at (3.25,0){$v_{5}$};
      \node (v6) at (2,0)[vertex]{};
      \node (v6l) at (2,-0.25){$v_{6}$};
      \draw[thick](v1)--(v2);
      \draw[thick](v2)--(v3);
      \draw[thick](v3)--(v4);                  
      \draw[thick](v4)--(v5);
      \draw[thick](v5)--(v6);
      \draw[thick](v6)--(v1);
      \draw[thick,dashed](v1)--(v3);
      \draw[thick,dashed](v1)--(v4);
      \draw[thick,dashed](v1)--(v5);
      \draw[thick,dashed](v2)--(v4);
      \draw[thick,dashed](v2)--(v5);
      \draw[thick,dashed](v2)--(v6);
      \draw[thick,dashed](v3)--(v5);
      \draw[thick,dashed](v3)--(v6);
      \draw[thick,dashed](v4)--(v6);
      \node (v7) at (1,2.5)[vertex]{};
      \node (v7l) at (0.75,2.5){$v_{7}$};
      \node (v8) at (4,2.5)[vertex]{};
      \node (v8l) at (4.25,2.5){$v_{8}$};
      \node (v9) at (1,-0.5)[vertex]{};
      \node (v9l) at (0.75,-0.5){$v_{9}$};
      \node (v10) at (4,-0.5)[vertex]{};
      \node (v10l) at (4.25,-0.5){$v_{10}$};
      \draw[thick](v7)--(v1);
      \draw[thick](v7)--(v2);
      \draw[thick](v8)--(v4);
      \draw[thick](v8)--(v5);
      \draw[thick,dotted](v8)--(v7);
      \draw[thick](v9)--(v1);
      \draw[thick](v9)--(v2);
      \draw[thick](v9)--(v3);
      \draw[thick](v9)--(v6);
      \draw[thick](v10)--(v3);
      \draw[thick](v10)--(v4);
      \draw[thick](v10)--(v5);
      \draw[thick](v10)--(v6);
      \draw[thick](v10)--(v9);
    \end{tikzpicture}  
  \end{center}
  \caption{An illustration of an XX-trigraphs with $X = \{v_{11},v_{12},v_{13}\}$. Between two vertices, a solid line indicates strong adjacency, a dashed line indicates (anti-)adjacency, a dotted line indicates possible adjacency, and no line indicates strong anti-adjacency. We already indicate that $v_1,\ldots,v_6$ form a cycle of strong adjacencies, as proved in Proposition~\ref{prp:prelim:xx}.}
\label{fig:xxtrigraphs}
\end{figure}

\begin{proposition}
\label{prp:prelim:xx}
  Let $G$ be a claw-free XX-trigraph obtained from the above graph on $\{v_{1},\ldots,v_{13}\}$ by removing $X \subseteq \{v_{7},v_{11},v_{12},v_{13}\}$.
  Then $\alpha(G) \leq 4$.
\end{proposition}
\begin{proof}
  We first claim that $v_{i}$ is strongly adjacent to $v_{i+1}$ for $i = 1,\ldots,5$ and $v_{6}$ is strongly adjacent to~$v_{1}$.
  To see this, observe that
  \begin{itemize}
    \item if $v_{1},v_{2}$ are semi-adjacent, then $v_{9},v_{1},v_{2},v_{10}$ is a claw;
    \item if $v_{2},v_{3}$ are semi-adjacent, then $v_{9},v_{2},v_{3},v_{6}$ is a claw;
    \item if $v_{3},v_{4}$ are semi-adjacent, then $v_{10},v_{3},v_{4},v_{6}$ is a claw;
    \item if $v_{4},v_{5}$ are semi-adjacent, then $v_{10},v_{4},v_{5},v_{9}$ is a claw;
    \item if $v_{5},v_{6}$ are semi-adjacent, then $v_{10},v_{3},v_{5},v_{6}$ is a claw;
    \item if $v_{6},v_{1}$ are semi-adjacent, then $v_{9},v_{1},v_{3},v_{6}$ is a claw.
  \end{itemize}
  The claim follows.

  Let $I$ be a maximum independent set of $G$.
  We can assume that $X = \emptyset$, because the size of a maximum independent set is non-increasing under the removal of vertices.
  Since adjacency between $v_{7},v_{8}$ can only decrease the size of a maximum independent set, we may assume that $v_{7},v_{8}$ are strongly anti-adjacent.
  Similarly, $v_{i}$ can be assumed to be strongly anti-adjacent to $v_{j}$ for each $i = 1,\ldots,4$ and each $i+2 \leq j \leq 6$.

We first show that we can assume that $v_{7},v_{8} \in I$.
  Observe that $v_{1}$, $v_{2}$, and $v_{13}$ are pairwise strongly adjacent. Hence, at most one of these vertices is in $I$. However, we also observe that $v_{7}$ is (strongly) adjacent to precisely $v_{1}$, $v_{2}$, and $v_{13}$. Therefore, if one of $v_{1},v_{2},v_{13}$ is in $I$, then $v_{7}$ is not. Moreover, we could replace the single vertex of $I \cap \{v_{1},v_{2},v_{13}\}$ by $v_{7}$. Conversely, if none of of $v_{1},v_{2},v_{13}$ is in $I$, then the choice of $I$ implies that $v_{7} \in I$. Thus, we can assume that $I \cap \{v_{1},v_{2},v_{13}\} = \emptyset$ and $v_{7} \in I$.
  Similarly, arguing about $v_{8}$ in relation to $v_{4},v_{5},v_{13}$, we can assume that $v_{8} \in I$ and $I \cap \{v_{4},v_{5},v_{13}\} = \emptyset$. Concluding, we have $v_{7},v_{8} \in I$ and  $v_{1},v_{2},v_{4},v_{5},v_{13} \not\in I$.

We now analyze all possible cases and prove that $|I| \leq 4$. Suppose that $v_{3} \in I$ or $v_{6} \in I$. Then $v_{9},v_{10},v_{11},v_{12} \not\in I$ by the definition of $G$. Hence, $I \subseteq \{v_{3},v_{6},v_{7},v_{8}\}$ and $|I| \leq 4$.
  So suppose that $v_{3},v_{6} \not\in I$. Recall that thus $v_{1},v_{2},v_{3},v_{4},v_{5},v_{6},v_{13} \not\in I$, but $v_{7},v_{8} \in I$.
  If $v_{9} \in I$ or $v_{10} \in I$, then $v_{11},v_{12} \not\in I$; hence, $I \subseteq \{v_{7},v_{8},v_{9},v_{10}\}$ and $|I| \leq 4$.
  If $v_{9},v_{10} \not\in I$, then $I \subseteq \{v_{7},v_{8}, v_{11}, v_{12}\}$ and $|I| \leq 4$.
  This exhausts all cases, and in each case $|I| \leq 4$.

Since $I$ is a maximum independent set of $G$, $\alpha(G) \leq 4$.
\qed
\end{proof}

The class $\mathcal{S}_{3}$ consists of all circular interval trigraphs.
A trigraph $G'$ is a \emph{(long) circular interval trigraph} if there exists a set of arcs $\mathcal{F} = \{F_{1},\ldots,F_{\ell}\}$ on a circle $\Sigma$, such that no two arcs share an endpoint (and the union of any three arcs does not contain the entire circle), and there is a mapping $p : V(G') \rightarrow \Sigma$ such that all points $p_{v'}$ are distinct and
\begin{itemize}
  \item $u',v'$ are adjacent if $p_{u'},p_{v'}$ are contained in a common arc $F_{i}$; moreover, if one of $p_{u'},p_{v'}$ is not an endpoint of $F_{i}$, then $u',v'$ are strongly adjacent;
  \item $u',v'$ are strongly anti-adjacent if there is no arc $F_{i}$ containing both $p_{u'},p_{v'}$.
\end{itemize}
Note that in the definition of a (long) circular interval trigraph the set of arcs is (long and) strict.
Moreover, one can assume that the set $\mathcal{F}$ of arcs is proper.

We now prove that any graph that is a thickening of a (long) circular interval trigraph is a fuzzy circular-arc graph.

\begin{lemma}
\label{lem:circ:circ-model}
  Any graph $G$ that is a thickening of a (long) circular interval trigraph has a representation such that each vertex $v$ of $G$ corresponds to an arc $A_{v}$, the set of arcs is (long,) almost proper and almost strict, and
  \begin{itemize}
    \item if $\{u,v\} \in E(G)$, then $A_{u}$ and $A_{v}$ intersect;
    \item if $\{u,v\} \not\in E(G)$, then either $A_{u}$ and $A_{v}$ do not intersect or $A_{u} \cap A_{v}$ consists of precisely one point.
  \end{itemize}
In particular, a thickening of a circular interval trigraph is a fuzzy circular-arc graph.
\end{lemma}
\begin{proof}
  Suppose that $G$ is a thickening $\mathcal{X} = \{X_{v'} \mid v' \in V(G')\}$ of a (long) circular interval trigraph $G'$. Let $p$ denote the mapping from $V(G')$ to the circle and let $\mathcal{F} = \{F_{1},\ldots,F_{\ell}\}$ denote the set of arcs that is (long,) proper and strict, as in the definition of a (long) circular interval trigraph.

We first show that $G'$ has a representation such that each vertex $v'$ of $G'$ corresponds to an arc $A_{v'}$, the set of arcs is (long,) proper and almost strict, and
  \begin{enumerate}
    \item $u',v'$ are strongly adjacent if and only if $A_{u'}$ and $A_{v'}$ intersect in more than one point;
    \item $u',v'$ are semi-adjacent if and only if $A_{u'} \cap A_{v'}$ contains exactly one point, which is an endpoint of $A_{u'}$ and an endpoint of $A_{v'}$;
    \item $u',v'$ are strongly anti-adjacent if and only if $A_{u'}$ and $A_{v'}$ do not intersect.
  \end{enumerate}
To this end, for each vertex $v' \in V(G')$, let $A_{v'}$ be the arc from $p_{v'}$ to the clockwise furthest endpoint~$r_{v'}$ of any arc $F_{i}$ containing $p_{v'}$.
If $r_{u'} = r_{v'}$ for distinct $u',v'\in V(G')$, then extend $A_{u'}$ or $A_{v'}$ by an infinitesimal amount clockwise if $p_{u'}$ appears after $p_{v'}$ or if $p_{v'}$ appears after $p_{u'}$, respectively. If, for distinct $u',v'\in V(G')$, there is an arc $F_{i}$ with endpoints $p_{u'}, p_{v'}$ where $r_{u'} = p_{v'}$ and $u',v'$ are strongly adjacent, then extend $A_{v'}$ by an infinitesimal amount counterclockwise.

We claim that the above construction gives the requested representation. The construction immediately implies that the set of arcs is long if $G'$ is a long circular interval trigraph. Furthermore, by the infinitesimal extensions that we performed, the set of arcs is proper, and moreover, also using that $\mathcal{F}$ is strict, the set of arcs is almost strict.
For the adjacency, we note that:
\begin{enumerate}
\item if $u'$ and $v'$ are strongly adjacent, then $p_{u'}$ or $p_{v'}$ is in the interior of an arc $F_{i}$, and thus by construction (and the second infinitesimal extension operation) the arcs $A_{u'}$ and $A_{v'}$ intersect in more than one point. The converse is argued similarly.
\item if $u'$ and $v'$ are semi-adjacent, then there is an arc $F_{i}$ with endpoints $p_{u'}$ and $p_{v'}$. Suppose that $p_{u'}$ comes clockwise before $p_{v'}$. Since $\mathcal{F}$ is proper, $r_{u'} = p_{v'}$ and moreover, $A_{v'}$ will not be extended infinitesimally counterclockwise by the second extension operation. Hence, $A_{u'} \cap A_{v'}$ contains exactly one point, which is an endpoint of $A_{u'}$ and an endpoint of $A_{v'}$. The converse is argued similarly.
\item if $u'$ and $v'$ are strongly anti-adjacent, then no arc $F_{i}$ contains both $p_{u'}$ and $p_{v'}$. By construction, $A_{u'}$ and $A_{v'}$ do not intersect. The converse is argued similarly.
\end{enumerate}
This proves the claim.

Finally, the representation as described in the lemma statement can be found by copying $A_{v'}$ for each $v \in X_{v'}$. 
The resulting set of arcs is clearly (long,) almost proper and almost strict. Moreover, strong (anti-)adjacency in $G'$ immediately implies (anti-)adjacency in $G$ by properties~1 and~3. Finally, semi-adjacency implies fuzzy adjacency by property~2.
\qed
\end{proof}

We also require the converse, i.e.~that any fuzzy circular-arc graph is a thickening of a circular interval trigraph.

\begin{lemma} \label{lem:circ:circ-model-back}
Any graph $G$ that has a representation as in Lemma~\ref{lem:circ:circ-model} (where the arcs are not necessarily long) is a thickening of a circular interval trigraph.
\end{lemma}
\begin{proof}
Let $\mathcal{A} = \{A_{v} \mid v \in V(G) \}$ be a representation for $G$ as in Lemma~\ref{lem:circ:circ-model}. We construct a circular interval trigraph $G'$. The set of vertices of $G'$ is constructed by taking the left endpoints of all arcs in $\mathcal{A}$. 
Let the set $\mathcal{F}$ be obtained from $\mathcal{A}$ by first removing all duplicate arcs; then, if two arcs share an endpoint, extend the clockwise second arc infinitesimally counterclockwise. Clearly, $G'$ is a circular interval trigraph. An appropriate thickening of $G'$ will bring back the duplicate arcs of $\mathcal{A}$ that were removed to obtain $\mathcal{F}$. Hence, $G$ is a thickening of $G'$.
\qed\end{proof}
We recall a result by Oriolo et al.~\cite{OrioloPS2012}, who showed that a thickening of a circular interval trigraph can be recognized in polynomial time. Combined with Lemma~\ref{lem:circ:circ-model} and Lemma~\ref{lem:circ:circ-model-back}, the result of Oriolo et al.~can be rephrased as follows.

\begin{theorem}[\cite{OrioloPS2012}] \label{thm:circ:fuzzy-reg}
Fuzzy circular-arc graphs can be recognized in polynomial time.
\end{theorem}

\subsection{Joins and Other Structures} \label{sec:trigraphs:joins}
Let $G$ be a trigraph throughout.
A \emph{$0$-join} is a partition of $V(G)$ into $V_{1},V_{2}$ such that $V_{1}$ is strongly anti-complete to $V_{2}$.

A \emph{$1$-join} is a partition of $V(G)$ into $V_{1},V_{2}$ for which there exist $A_{1} \subseteq V_{1}$ and $A_{2} \subseteq V_{2}$ such that
\begin{itemize}
  \item $A_1 \cup A_2$ is a strong clique;
  \item $V_1 \setminus A_1$ is strongly anti-complete to $V_{2}$ and $V_{2} \setminus A_{2}$ is strongly anti-complete to $V_{1}$;
  \item $A_{1}$, $A_{2}$, $V_{1}\setminus A_{1}$, and $V_{2} \setminus A_{2}$ are nonempty.
\end{itemize}
The definition of a \emph{pseudo-$1$-join} is the same as that of a $1$-join, except that the last condition is replaced by the condition that $V_{1},V_{2}$ are not allowed to be strong independent sets.
Observe that a trigraph admitting a $1$-join but not a $0$-join also admits a pseudo-$1$-join.

A \emph{generalized $2$-join} is a partition of $V(G)$ into $V_{0},V_{1},V_{2}$ for which there exist disjoint sets $A_{1},B_{1} \subseteq V_{1}$ and $A_{2},B_{2} \subseteq V_{2}$ such that
\begin{itemize}
  \item $A_{1} \cup A_{2} \cup V_{0}$ and $B_{1} \cup B_{2} \cup V_{0}$ form a strong clique;
  \item $V_{1} \setminus (A_{1} \cup B_{1})$ is strongly anti-complete to $V_{0} \cup V_{2}$ and $V_{2} \setminus (A_{2} \cup B_{2})$ is strongly anti-complete to $V_{0} \cup V_{1}$;
  \item $A_{1}$, $A_{2}$, $B_{1}$, $B_{2}$, $V_{1}\setminus (A_{1} \cup B_{1})$, and $V_{2} \setminus (A_{2} \cup B_{2})$ are nonempty.
\end{itemize}
The definition of a \emph{pseudo-$2$-join} is the same as that of a generalized $2$-join, except that the last condition is replaced by the condition that $V_{1},V_{2}$ are not allowed to be strong independent sets.
Observe that a trigraph admitting a generalized $2$-join but not a $0$-join also admits a pseudo-$2$-join.

Disjoint sets $A,B \subseteq V(G)$ form a \emph{W-join} if $A$ and $B$ are strong cliques, every vertex of $V(G)\setminus(A \cup B)$ is either strongly complete or strongly anti-complete to $A$ and either strongly complete or strongly anti-complete to $B$, $A$ is neither strongly complete nor strongly anti-complete to $B$, and~$\max\{|A|,|B|\} \geq 2$.

The trigraph $G$ admits \emph{twins} if it has two strongly adjacent vertices for which any third vertex is either strongly adjacent or strongly anti-adjacent to both vertices.
A set of vertices that are pairwise twins is called a \emph{twin set}.

The definition of a strip-structure extends naturally to trigraphs, where we demand that $C(h)$ is a strong clique, and that for any strip $(J,Z)$, $Z$ is a strong independent set and any vertex of $V(J)$ that is adjacent to some vertex of $Z$ is in fact strongly adjacent to it.
Similarly, we call a strip $(J,Z)$ a spot if $J$ consists of three vertices $z_{1},s,z_{2}$, such that $z_{1},z_{2}$ are strongly anti-adjacent and both $z_{1}$ and $z_{2}$ are strongly adjacent to $s$.
The definition of a stripe again extends directly.
Then a strip-structure of a trigraph is called \emph{purified} if each strip is either a spot or a stripe.

Finally, a stripe $(J,Z)$ is a \emph{thickening of a stripe} $(J',Z')$ if $J$ is a thickening $\mathcal{X}$ of $J'$ and there is a bijection $\sigma$ between $Z$ and $Z'$ such that $X_{z'} = \{\sigma(z')\}$.

\subsection{Supporting Results} \label{sec:trigraphs:end}
We need the following auxiliary results from Hermelin~et al.~\cite{HermelinMLW11} before we can prove Theorem~\ref{thm:prelim:main}.

Call a stripe $(J,Z)$ \emph{almost-unbreakable} if $J$ does not admit a $0$-join, a pseudo-$1$-join, or a pseudo-$2$-join, $J[V(J) \setminus Z]$ does not admit twins, and $J$ does not admit a W-join $(A,B)$ such that $Z \cap (A \cup B) = \emptyset$.

\begin{lemma}[\cite{HermelinMLW11}]
\label{lem:prelim:strip-structure}
  Every connected claw-free trigraph admits a purified strip-structure in which its strips are either spots or thickenings of almost-unbreakable stripes.
  Moreover, if $G$ is a graph, such a strip-structure can be found in polynomial time.
\end{lemma}

\begin{lemma}[\cite{HermelinMLW11}]
\label{lem:prelim:unbreakable-structure}
  Let $(J,Z)$ be (a thickening of) an almost-unbreakable stripe.
  Then either $(J,Z)$ is a thickening of a stripe $(J',Z')$ such that $J'$ is a member of one of $\mathcal{S}_{0},\ldots,\mathcal{S}_{7}$, or $J$ is a union of at most three strong cliques and $|Z| \leq 2$.
\end{lemma}

\begin{lemma}[\cite{HermelinMLW11}]
\label{lem:prelim:z-bound}
Let $(J,Z)$ be an almost-unbreakable stripe such that $Z \not= \emptyset$ and $J$ is a thickening of a member of one of $\mathcal{S}_{1},\ldots,\mathcal{S}_{7}$. Then $|Z| \leq 2$.
\end{lemma}

\begin{lemma}[\cite{HermelinMLW11}]
\label{lem:prelim:line-graph}
  Let a graph $G$ be a thickening of a line trigraph such that $G$ admits no $0$-join, pseudo-$1$-join, or pseudo-$2$-join.
  Then $G$ is the line graph of a multigraph without self-loops, or $G$ is a union of two strong cliques.
\end{lemma}

\begin{lemma}[\cite{HermelinMLW11}] \label{lem:prelim:alpha}
  Let $G$ be a trigraph such that $G \in \mathcal{S}_{1} \cup \mathcal{S}_{4} \cup \mathcal{S}_{5} \cup \mathcal{S}_{6} \cup \mathcal{S}_{7}$.
  Then $\alpha(G) \leq 3$.
\end{lemma}
We require the following observation, which follows from the definition of a thickening.

\begin{proposition}
\label{prp:prelim:alpha-thickening}
  Let $G$ be a trigraph that is a thickening of a trigraph $G'$.
  Then $\alpha(G) = \alpha(G')$.
\end{proposition}
\begin{proof}
We prove a slightly stronger statement, namely that each independent set of $G'$ corresponds to an independent set of $G$ of equal size, and vice versa. Throughout, let $G$ be obtained from $G'$ using a thickening $\mathcal{X} = \{X_{v'} \mid v' \in V(G')\}$.

Let $I'$ be an independent set of $G'$. To construct an independent set $I$ of $G$, consider each $v' \in I'$ in turn. If $v' \in I'$ is semi-adjacent to some $u' \in I'$, then from the definition of thickening there exist an anti-adjacent pair $v \in X_{v'}$ and $u \in X_{u'}$; if $u'$ has not been considered yet, add both $u$ and $v$ to $I$. Otherwise, i.e.~if $v' \in I$ is not adjacent to any $u' \in I'$, then add an arbitrary vertex of $X_{v'}$ to $I$. By construction, $|I| = |I'|$. Moreover, since $G'$ is a trigraph, $v'$ is semi-adjacent to at most one vertex of $G'$. Therefore, it follows from the definition of thickening that $I$ is an independent set of $G$.

Let $I$ be an independent set of $G$. Let $I'$ contain $v'$ if and only if $I \cap X_{v'} \not= \emptyset$. Since $X_{v'}$ is a strong clique for each $v' \in V(G')$, $|I \cap X_{v'}| \leq 1$ for each $v' \in V(G')$, and thus $|I'| = |I|$. Moreover, by the definition of thickening, $u,v \in V(G)$ are anti-adjacent only if $u \in X_{u'}$, $v \in X_{v'}$, and $u',v'$ are anti-adjacent, for some $u',v' \in V(G')$. Hence, $I'$ is an independent set of $G'$.
\qed\end{proof}

We also need the following observation.

\begin{proposition}\label{prp:prelim:pseudo}
Let $G$ be a trigraph that is a thickening of a trigraph $G'$ such that $G$ does not admit a $0$-join, a pseudo-$1$-join, or a pseudo-$2$-join. Then $G'$ does not admit a $0$-join, a pseudo-$1$-join, or a pseudo-$2$-join.
\end{proposition}
\begin{proof}
Let $G$ be obtained from $G'$ using a thickening $\mathcal{X} = \{X_{v'} \mid v' \in V(G')\}$.
If $G'$ admits a $0$-join $V_{1},V_{2}$, where $V_{1},V_{2}$ are as in the definition of a $0$-join, then $V_{1}'= \bigcup_{v' \in V_{1}} X_{v'}$ and $V_{2}' = \bigcup_{v' \in V_{1}} X_{v'}$ is a $0$-join of $G$, a contradiction. A similar observation holds in the case that $G'$ would admit a pseudo-$1$-join or a pseudo-$2$-join.
\qed\end{proof}

\subsection{Proof of the Algorithmic Decomposition Theorem} \label{sec:trigraphs:main}
The proof of the decomposition theorem is similar in structure to the proof of the decomposition in Hermelin~et al.~\cite{HermelinMLW11}.

\begin{proof}[of Theorem~\ref{thm:prelim:main}]
By Lemma~\ref{lem:prelim:strip-structure}, $G$ has a purified strip-structure $(\mathcal R, \{(J_{e}, Z_{e})\} \mid e \in E(\mathcal R)\})$ such that its strips are connected and are spots or thickenings of almost-unbreakable stripes. Moreover, since $G$ is a graph, we can compute such a strip-structure in polynomial time.
We distinguish several cases.

First, suppose that $|E(\mathcal R)| = 1$. Then the strip-structure is the trivial strip-structure consisting of a single strip $(G,\emptyset)$. Note that $(G,\emptyset)$ is not a spot by definition. Therefore, $G$ is a thickening of an almost-unbreakable stripe. Since $\alpha(G) > 4$, $G$ is not a union of at most three strong cliques. Then Lemma~\ref{lem:prelim:unbreakable-structure} implies that $G$ is a thickening of a trigraph $G'$ that is a member of one of $\mathcal{S}_{0},\ldots,\mathcal{S}_{7}$. Since $\alpha(G) > 4$, Proposition~\ref{prp:prelim:alpha-thickening} implies that $\alpha(G') > 4$. Then, by Lemma~\ref{lem:prelim:alpha} and Proposition~\ref{prp:prelim:xx}, $G' \in \mathcal{S}_{0} \cup \mathcal{S}_{3}$. Since $G$ is not a fuzzy circular-arc graph by assumption, Lemma~\ref{lem:circ:circ-model} implies that $G' \not\in \mathcal{S}_{3}$. Hence, $G' \in \mathcal{S}_{0}$. It remains to output a strip-structure that satisfies the conditions of the theorem statement.

We construct a special strip-structure in the remaining case that $G' \in \mathcal{S}_{0}$. Since $G'$ admits no $0$-join, pseudo-$1$-join, or pseudo-$2$-join, and $\alpha (G') > 4$, Lemma~\ref{lem:prelim:line-graph} implies that $G'$ is the line graph of a multigraph without self-loops. Since $G'$ is a graph, any thickening of $G'$ can also be obtained by replacing each vertex of $G'$ with a twin set. Since line graphs are closed under adding twins, $G$ is also the line graph of a multigraph without self-loops, which we denote by $M$. As the first step, we find $M$ in linear time.
  First, we find all twin sets in linear time using an algorithm implicit in Habib et al.~\cite{HabibMPV2000}.
  Second, we remove all but one vertex from each twin set and mark this remaining vertex with the size of the original twin set.
  Third, we run the linear-time algorithm by Roussopoulos~\cite{Roussopoulos1973} on this marked graph $G'$ (which must be a line graph) to find its pre-image.
  Finally,~$M$ is obtained by duplicating all edges in this pre-image according to the markings of the corresponding vertices in $G'$.
  As the second step, we use $M$ to compute the required strip-structure $\mathcal{R}$ with strips $(J_{e},Z_{e})$.
  First, we choose
  \begin{equation*}
  V(\mathcal{R}) = V(M) \setminus \{ v \in V(M) \mid v\ \mbox{is a pendant vertex in}\ M\},
  \end{equation*}
  and
  \begin{eqnarray*}
    E(\mathcal{R}) & = & \{\, \{u,v\} \mid \{u,v\} \in E(M)\ \mbox{and neither of}\ u,v\ \mbox{is pendant in}\ M\} \ \cup\\
    && \{\, \{u\} \mid \{u,v\} \in E(M)\ \mbox{and}\ v\ \mbox{is a pendant vertex in}\ M\}.
  \end{eqnarray*}
  Note that each strip-edge of $\mathcal{R}$ corresponds to an edge of $M$.
  Then for each $e \in E(\mathcal{R})$, let $w_{e}$ denote the vertex of $G$ corresponding to the edge of $M$ that corresponds to $e$. We can now define the strips. For $e \in E(\mathcal{R})$, let $Z_{e}$ be a set of $|e|$ new vertices, and let $J_{e}$ be the graph consisting of $w_{e}$ and the vertices in $Z_{e}$ such that $w_{e}$ is (strongly) adjacent to all vertices of $Z_{e}$. We observe that this is indeed a strip-structure of $G$. Moreover, each strip is a spot or a stripe where a maximum independent set has size one (and thus at most four).

Suppose that $|E(\mathcal R)| > 1$.
We consider each strip-edge $e \in E(\mathcal{R})$ and its associated strip $(J_{e},Z_{e})$ in turn, and check whether further simplification is necessary. To improve legibility, we drop the subscript $e$ and let $(J,Z)$ denote the considered strip. We also use $r_{z}$ to denote the unique strip-vertex in $e$ for which $z = z^{r}_{e}$, for each $z \in Z$.

We can check in constant time whether~$(J,Z)$ is a spot. Since spots are part of the theorem statement, we do not need to simplify and can proceed to the next strip.
Then, $(J,Z)$ is a thickening of an almost-unbreakable stripe $(J',Z')$. 
Since $G$ is connected, $|Z| = |Z'| \geq 1$.
We can check in polynomial time whether $\alpha(J) \leq 4$ and $|Z| \leq 2$. Since such stripes are part of the theorem statement, we do not need to simplify and can proceed to the next strip.
We can also check in polynomial time whether $J$ is a fuzzy circular-arc graph (see Theorem~\ref{thm:circ:fuzzy-reg}) and $|Z| \leq 2$. Since such stripes are part of the theorem statement, we do not need to simplify and can proceed to the next strip.
Now, if $(J',Z')$ is a union of at most three strong cliques and $|Z'| \leq 2$, then $|Z| \leq 2$ and $\alpha(J) \leq 4$ by Proposition~\ref{prp:prelim:alpha-thickening}, a contradiction. 
Then Lemma~\ref{lem:prelim:unbreakable-structure} implies that $(J',Z')$ is a thickening of a stripe $(J'',Z'')$ such that $J''$ is a member of one of $\mathcal{S}_{0},\ldots,\mathcal{S}_{7}$.
The assumptions on $(J',Z')$, Lemma~\ref{lem:prelim:z-bound}, Lemma~\ref{lem:prelim:alpha}, Proposition~\ref{prp:prelim:xx}, and Proposition~\ref{prp:prelim:alpha-thickening} imply that $J''$ is not a member of one of $\mathcal{S}_{1},\mathcal{S}_{2}, \mathcal{S}_{4},\ldots,\mathcal{S}_{7}$.
Furthermore, the assumptions on $(J',Z')$, Lemma~\ref{lem:circ:circ-model}, and Lemma~\ref{lem:prelim:z-bound} imply that $J''$ is not a member of $\mathcal{S}_{3}$.
Therefore, $J''$ is a member of $\mathcal{S}_{0}$.

Since $J''$ is a member of $\mathcal{S}_{0}$, we need to simplify further. As $J$ is a thickening of $J'$ and $J''$ is a thickening of $J'$, it follows from the definition of thickening that $J$ is a thickening of $J''$. Since $J'$ is almost-unbreakable, it follows from Proposition~\ref{prp:prelim:pseudo} that $J''$ does not admit a $0$-join, a pseudo-$1$-join, or a pseudo-$2$-join. Therefore, by Lemma~\ref{lem:prelim:line-graph}, $J$ is the line graph of a multigraph $M$. As outlined before, we can compute $M$ in polynomial time.

We first apply a minor modification to $M$. Consider the edge $f_{z}$ of $M$ that corresponds to some $z \in Z$, and let $a,b$ denote its endpoints. Suppose that $a$ is adjacent to a vertex $c_{a} \not= b$ in $M$, and that $b$ is adjacent to a vertex $c_{b} \not= a$ in $M$. Since the neighborhood of $z$ in $J$ is a strong clique by the definition of a strip, $a$ and $b$ are incident to exactly one and the same vertex $c$ of $M$, and in particular, $c = c_{a} = c_{b}$. Therefore, we can replace each edge $\{b,c\}$ of $M$ by an edge $\{a,c\}$, and $J$ is still the line graph of this modified graph. We apply now apply this modification exhaustively, and by abuse of notation, call the resulting multigraph $M$ as well. Observe that now $f_{z}$ contains a pendant vertex (i.e.~a vertex adjacent to exactly one other vertex) for each $z \in Z$.

As outlined before, we can find a strip-structure for $M$ in polynomial time. Let $\mathcal{R}'$ denote this strip-structure, and let $(J'_{e'},Z'_{e'})$ denote the strip associated with each $e' \in E(\mathcal{R'})$. Note that each $v \in J$ corresponds to some strip-edge $e' \in E(\mathcal{R'})$, which we denote by $e'_{v}$. We now integrate $\mathcal{R}'$ into $\mathcal{R}$. Consider each $z \in Z$ in turn. By our modification step and the way that we construct $\mathcal{R}'$, $|e'_{z}| = 1$. Let $r'_{z}$ denote the single strip-vertex in $e'_{z}$. Then, as $Z$ is a strong independent set, no strip-edge of $\mathcal{R}'$ incident on $r'_{z}$ can correspond to another vertex of $Z$. Therefore, we can remove $e'_{z}$ and replace each occurrence of $r'_{z}$ in a strip-edge of $\mathcal{R}'$ with $r_{z}$. We do this for each $z \in Z$. Then, finally, add all remaining strip-vertices and strip-edges of $\mathcal{R}'$ to $\mathcal{R}$, and remove $e$ from $\mathcal{R}$. One can readily verify that this modification of $\mathcal{R}$ preserves all five properties of Definition~\ref{def:strip-structure}, and thus the modified $\mathcal{R}$ and associated strips forms a strip-structure for $G$. Moreover, we have removed the offending strip $(J,Z)$ and replaced it with strips that are either spots or stripes where a maximum independent set has size one (and thus at most four).

We observe that all above computations take polynomial time in total, and that the resulting strip-structure has the properties set forth in the theorem statement.
\qed
\end{proof}

\section{Fixed-Parameter Algorithm on Claw-Free Graphs}
\label{sec:algo}
In this section, we provide our main result, a fixed-parameter algorithm for {\sc Induced Graph Matching} on claw-free graphs, parameterized by the size $k$ of the matching and for any fixed connected graph $H$.
Throughout the section, we use $H$ to denote a connected graph on $h$ vertices, and $G$ to denote a claw-free graph on $n$ vertices, so that $G$, $H$, and an integer $k\in\mathbb N$ are given as input to the {\sc Induced Graph Matching} problem.
For simplicity, we assume that $G$ is connected, as the extension to the disconnected case is immediate.

Our algorithm deploys the claw-free decomposition theorem stated in Theorem~\ref{thm:prelim:main}.
Note that in order to apply this decomposition theorem, we need to first handle two cases: the case where $\alpha(G)$ is small, and the case where $G$ is a fuzzy circular-arc graph.

\begin{proposition}
\label{prp:algo:alpha}%
  {\sc Induced Graph Matching} on graphs $G$ with $\alpha(G) \leq 4$ can be solved in $n^{O(h)}$ time.
\end{proposition}
\begin{proof}
  Note that since $\alpha(G) \leq 4$, any induced $H$-matching of $G$ will have size at most four.
  Hence, by exhaustive enumeration, {\sc Induced Graph Matching} can be solved in $n^{O(h)}$ time.
\qed
\end{proof}

By Theorem~\ref{thm:algo:circ} and Proposition~\ref{prp:algo:alpha}, we can solve {\sc Induced Graph Matching} in polynomial time if $G$ is a fuzzy circular-arc graph or if $\alpha(G) \leq 4$.
Hence, we may assume from now on that $G$ is not a fuzzy circular-arc graph and satisfies $\alpha(G) > 4$. We can then apply Theorem~\ref{thm:prelim:main} to obtain a strip-structure of $G$. Let $\mathcal R$ denote the strip-graph of this strip-structure.

The remainder of this section consists of two parts. In the first part (Section~\ref{sec:algo:covered}), we exhibit a global structural impression that an induced $H$-matching of size $k$ would leave on $G$ and, in particular, on $\mathcal{R}$. In the second part (Section~\ref{sec:algo:color}), we turn this around, in that we try to find certain structures in $G$ and, in particular, in $\mathcal{R}$, and then use this to discover an induced $H$-matching of size $k$ (if it exists).

\subsection{Covered Subgraph and Annotations} \label{sec:algo:covered}
In this subsection, we describe the structure that an induced $H$-matching of size $k$ would impose on $G$ and $\mathcal{R}$. We first formalize the notion of the covered subgraph of $\mathcal R$.

\begin{definition} \label{def:algo:covered}
Let $M$ be an induced $H$-matching of $G$.
We call a strip-edge of $\mathcal R$ \emph{covered} if its corresponding strip contains a vertex of $M$.
Similarly, we call a strip-vertex $r$ of $\mathcal R$ \emph{covered} if the clique $C(r)$ corresponding to $r$ contains a vertex of $M$.
The \emph{covered subgraph} $C_{M}$ of $\mathcal{R}$ is formed by the covered strip-edges and the strip-vertices that they contain. 
\end{definition}
Note that by definition, all covered strip-vertices are part of the covered subgraph.
By straightforward extension of the definition, we may also speak of the covered subgraph of an occurrence of $H$.

\begin{proposition} \label{prp:algo:covered}
If $G$ has an induced $H$-matching $M$ of $G$ of size $k$, then $C_{M}$ has at most $hk$ strip-edges and $2hk$ strip-vertices.
\end{proposition}
\begin{proof}
It suffices to observe that each vertex of an occurrence of $H$ in $M$ is contained in a strip, which corresponds to some strip-edge. Since each strip-edge contains at most two strip-vertices, the proposition follows.
\qed\end{proof}
The covered subgraph gives a global impression of the structure that an induced $H$-matching of size at least $k$ induces in $\mathcal{R}$. We now need to obtain a global impression of which vertices of an induced $H$-matching are in which strip-edge and whether they are in a boundary or in the interior of a strip. To this end, create a set of \emph{tokens}: $k$ tokens for every vertex of $H$, leading to $hk$ tokens. We can partition these tokens into $k$ \emph{token groups} of size $h$, which each contain one token of each vertex of $H$. The tokens can be seen as a representation of the vertices of $k$ occurrences of $H$. We use them in the following definition:

\begin{definition} \label{def:algo:annotation}
An \emph{annotation} of a hypergraph with edges of size one or two describes for each edge $e$ of the hypergraph:
\begin{itemize}
\item whether $e$ corresponds to a stripe or (if $e$ contains two vertices) a spot;
\item if $e$ corresponds to a stripe, then which tokens of the $hk$ tokens are assigned to the interior of the stripe and, for each boundary of the stripe, which tokens of the $hk$ tokens are assigned to that boundary;
\item if $e$ corresponds to a spot, then which token of the $hk$ tokens is assigned to the spot.
\end{itemize}
All $hk$ tokens should be assigned in this way.
\end{definition}
Observe that in a way, an induced $H$-matching of $G$ of size $k$ is nothing more than a constrained way to assign the $hk$ tokens to vertices of $G$. The following annotation of the covered subgraph then follows.

\begin{proposition} \label{prp:algo:annotation}
If $G$ has an induced $H$-matching $M$ of $G$ of size $k$, then $M$ induces an annotation of $C_{M}$.
\end{proposition}
\begin{proof}
We construct the annotation as follows. Let $e$ be an edge of $C_{M}$. Note that $C_{M}$ already contains the information to annotate whether $e$ corresponds to a stripe or a spot. Now consider $M$ as an assignment of the $hk$ tokens to the vertices of $G$. If $e$ corresponds to a stripe, then augment the annotation according to which tokens of $M$ are assigned to vertices in the interior of the stripe and, for each boundary of the stripe, which tokens of $M$ were assigned to vertices of that boundary. If $e$ corresponds to a spot, then only a single token of $M$ (if any) can be assigned to its vertex. Therefore, we can augment the annotation by assigning this token to $e$. Since $M$ is an induced $H$-matching of size $k$, all tokens are assigned.
\qed\end{proof}
We call the annotation described in Proposition~\ref{prp:algo:annotation} the \emph{natural annotation} of $C_{M}$.

\subsection{Algorithm} \label{sec:algo:color}
The previous subsection shows (roughly) that an induced $H$-matching of size $k$ imposes a subgraph of $\mathcal{R}$ with $hk$ strip-edges and $2hk$ strip-vertices (Proposition~\ref{prp:algo:covered}) and an annotation of that subgraph (Proposition~\ref{prp:algo:annotation}). We now attempt the converse: we try to find such a subgraph and an annotation of it, and use this to find an induced $H$-matching of size $k$ (if it exists). A major complication to this idea is that it seems to require solving an instance of {\sc Subgraph Isomorphism}, which is known to be $\mathsf{W[1]}$-hard when parameterized by the size of the pattern. We can, however, get around this seeming complication by taking advantage of the color-coding technique. We construct colors for each vertex of the guessed covered subgraph and, for each edge, colors for the different items of the guessed annotation. Then we use the color-coding technique~\cite{AlonYZ1995} to distribute these colors on $\mathcal{R}$, so that each strip-vertex and strip-edge receives some color. Finally, we try to recover as much of the guessed covered subgraph and its annotation from this coloring of $\mathcal{R}$, and ultimately find an induced $H$-matching of size $k$ (if it exists).

We formalize the above intuitive description. The algorithm consists of five major steps, which we describe and analyze in turn, and then give a proof of correctness as well as analysis of the running time.

\paragraph{Step~1:~Bases}
We call a \emph{base} the combination of a hypergraph on at most $hk$ edges and $2hk$ vertices and an annotation of this hypergraph. Two bases are distinct if their hypergraphs are non-isomorphic\footnote{We note that the definition of isomorphism can be easily extended to hypergraphs.} or if their hypergraphs are isomorphic, then their annotations are distinct.

\begin{proposition} \label{prp:algo:base}
The number of distinct bases is $(hk)^{O(hk)}$ and they can be enumerated in $(hk)^{O(hk)}$ time.
\end{proposition}
\begin{proof}
A trivial upper bound on the number of hypergraphs on at most $2hk$ vertices and at most $hk$ edges of size one or two is $(hk)^{O(hk)}$ --- think that $hk$ edges choose at most two of $2hk$ vertices, where empty edges are discarded. Since each of the hypergraphs has at most $hk$ edges, the number of possible annotations of such a hypergraph is bounded by $2^{hk}\, (3hk)^{hk}$. The proposition follows.
\qed\end{proof}
The following proposition is immediate from Proposition~\ref{prp:algo:covered} and~\ref{prp:algo:annotation}.

\begin{proposition} \label{prp:algo:base-k}
If $G$ has an induced $H$-matching of size $k$, then there is a base such that the base hypergraph is isomorphic to the covered subgraph and the base annotation is equal to the natural annotation of the covered subgraph.
\qed\end{proposition}
We now enumerate all bases using Proposition~\ref{prp:algo:base}. Let $B$ denote the hypergraph of the current base.

We then verify whether the base satisfies two necessary (but not sufficient) conditions for it to be equivalent to the natural annotation of the covered subgraph of an induced $H$-matching.

\begin{proposition}[Condition~1] \label{prp:algo:cond1}
If two vertices of $H$ are adjacent, then the corresponding tokens of each token group must be assigned to the same edge of $B$ or to two different edges of $B$ that share an endpoint.
In the latter case, the annotation must also assign the tokens to the boundaries that correspond to that endpoint.
\end{proposition}
\begin{proof}
By the definition of a strip-structure (recall Definition~\ref{def:strip-structure}), two vertices of $G$ can only be adjacent if they are part of the same strip-edge of $\mathcal{R}$ or part of two different strip-edges of $\mathcal{R}$ that share an endpoint. In the latter case, they also have to be part of the boundaries that correspond to that endpoint.
Hence, if $G$ has an induced $H$-matching $M$ of size $k$, $B$ is isomorphic to the covered subgraph $C_{M}$, and the base annotation is equal to the natural annotation of $C_{M}$, then the proposition statement is indeed true.
\qed\end{proof}

\begin{proposition}[Condition~2] \label{prp:algo:cond2}
For each vertex $b$ of $B$, the boundaries corresponding to~$b$ of the edges incident to $b$ can contain tokens of at most one token group. Moreover, the vertices of $H$ corresponding to these tokens must induce a clique in~$H$.
\end{proposition}
\begin{proof}
By the definition of a strip-structure, the vertices of any boundary induce a clique in $G$. Therefore, in any induced $H$-matching of $G$, at most one occurrence of $H$ can have vertices in each boundary and these vertices form a clique. Hence, if $G$ has an induced $H$-matching $M$ of size $k$, $B$ is isomorphic to the covered subgraph $C_{M}$, and the base annotation is equal to the natural annotation of $C_{M}$, then the proposition statement is indeed true.
\qed\end{proof}
If $B$ does not satisfy either Condition~1 or~2, then we proceed to the next base. Hence, from now on, we assume that $B$ satisfies Condition~1 and~2.

\paragraph{Step~2:~Color Coding}
We now apply the technique of color-coding~\cite{AlonYZ1995}.
We use a set of colors $\Gamma_{B}$ that consists of:
\begin{itemize}
\item a unique color for each vertex of $B$;
\item a unique color for each edge of $B$ that is annotated as a spot by the base annotation (a \emph{spot color});
\item at most three unique colors for each edge of $B$ that is annotated as a stripe by the base annotation: one for the interior of the stripe (an \emph{interior color}) and one for each boundary (a \emph{boundary color}).
\end{itemize}
The colors for the vertices of $B$ are called \emph{vertex colors} and the colors for the edges of $B$ are called \emph{edge colors}.

Next, we describe the set $\mathcal E$ of elements that we wish to color. The set $\mathcal{E}$ consists of:
\begin{itemize}
\item an element for each strip-vertex of $\mathcal R$;
\item an element for each strip-edge of $\mathcal R$ that corresponds to a spot;
\item at most three elements for each strip-edge of $\mathcal R$ that corresponds to a stripe: an element for the interior of the stripe (an \emph{interior element)} and an element for each boundary (a \emph{boundary element}).
\end{itemize}
The elements for the strip-vertices of $\mathcal R$ are called \emph{vertex elements} and the elements for the strip-edges of $\mathcal{R}$ are called \emph{edge elements}. 

We call a \emph{$\Gamma_{B}$-coloring} any coloring of $\mathcal E$ with the colors of $\Gamma_{B}$. The following proposition is immediate from the work of Alon~et al.~\cite{AlonYZ1995} on perfect hash families.

\begin{proposition} \label{prp:algo:colorcoding}
There exists a family $\mathcal{F}_{B}$ of $\Gamma_{B}$-colorings such that for any subset $X$ of $|\Gamma_B|$ elements in $\mathcal E$ there is a subset of colorings in $\mathcal{F}_{B}$ that will assign distinct colors to $X$ in all possible ways. Moreover, the family $\mathcal{F}_{B}$ has size $f(k,h) \cdot n^{O(1)}$ and can be computed in $g(k,h) \cdot n^{O(1)}$ time, where $f,g$ are computable functions.
\qed\end{proposition}
We now consider each $\Gamma_{B}$-coloring in $\mathcal{F}_{B}$ in turn. Let $F$ denote the current $\Gamma_{B}$-coloring.

\paragraph{Step~3:~Blanking}
Ideally, from the coloring $F$, we would like to recover a subgraph of $\mathcal R$ that is isomorphic to~$B$. However, as discussed before, this is unlikely to succeed using a fixed-parameter algorithm. Instead, we look to obtain a (partial) surjection from $\mathcal{R}$ onto $B$, which is a necessary (but not sufficient) condition. Therefore, we remove the color of any vertex or edge element if the color is not consistent with this goal. A vertex or edge element without color is called \emph{blanked}, whereas an element that still has a color is called \emph{unblanked}.
To simplify the description, we will not distinguish between a strip-edge (strip-vertex) of $\mathcal{R}$ and its corresponding edge elements (vertex element). In particular, by blanking a strip-edge, we mean removing the color of all edge elements that correspond to that strip-edge.

To describe the blanking procedure, we need one further definition, namely the \emph{color sequence} of an edge $f$ of $B$. We distinguish three cases:
\begin{itemize}
\item If $f$ is on a vertex $p$, then the color sequence of $f$ is the vertex color of $p$, the boundary color of $f$, and the interior color of $f$ (in this order).
\item If $f$ is on two vertices $p,q$ and is annotated as a spot, then the color sequence of $f$ is the vertex color of $p$, the spot color of $f$, and the vertex color of $q$ (in this order).
\item If $f$ is on two vertices $p,q$ and is annotated as a stripe, then the color sequence of $f$ is the vertex color of $p$, the boundary color of $f$ for the boundary that corresponds to $p$, the interior color of $f$, the boundary color of $f$ for the boundary that corresponds to $q$, and the vertex color of $q$ (in this order).
\end{itemize}
We can construct a similar order on the elements of each strip-edge $e$ of $\mathcal{R}$; we call this the \emph{element sequence} of $e$.
Then we blank:
\begin{enumerate}
\item any strip-vertex that has not received a vertex color;
\item any strip-edge for which the ordered colors of its element sequence do not form a color sequence.
\end{enumerate}
Rule~1 is just common sense. Rule~2 has four important consequences. First, it ensures that each strip-edge is colored as expected, in the sense that a boundary element receives a boundary color, etc. Second, we can use the coloring to establish a surjection $\delta$ of unblanked strip-vertices and strip-edges of $\mathcal R$ onto vertices and edges of $B$. Third, for any unblanked strip-edge $e$ that contains a strip-vertex $x$, $\delta(e)$ contains $\delta(x)$ in $B$. In particular, the color of an unblanked strip-vertex determines the allowed colors for the unblanked strip-edges that contain the strip-vertex. Fourth, for any unblanked strip-vertex $x$, the boundary color of the boundary corresponding to $x$ of any strip-edge $e$ that contains $x$ is equal to the boundary color of the boundary of $\delta(e)$ corresponding to $\delta(x)$. In other words, the boundaries `line up' with the strip-vertices.

Let $F'$ denote the resulting coloring of $\mathcal E$ and let $\delta$ denote the resulting surjection. We note that if $\delta$ would not be a surjection (i.e.~there is a color of $\Gamma_{B}$ that does not appear in $F'$), then we can immediately proceed to the next coloring in $\mathcal{F}_{B}$.

\paragraph{Step~4:~Strip Interiors}
In this step, we want to find a largest number of occurrences of $H$ that are in the interior of a strip, conditioned on what the coloring says the boundary of the strip should look like.
To this end, consider all strip-edges $e$ that are unblanked in turn. Let $(J,Z)$ denote the strip corresponding to $e$.
We assume that $(J,Z)$ is a stripe with $|Z| = 2$; the other cases are similar.
Using the colors of $\delta(e)$, we can find the tokens assigned to the boundaries of the strip by the annotation.
Recall that by Condition~2, for each boundary, tokens of at most one token group are assigned.
Let $T^{1}$ denote the set of all tokens assigned by the annotation to $\delta(e)$ that belong to the token group that has a token assigned to the first boundary of the strip.
We define $T^{2}$ similarly with respect to the second boundary (note that possibly $T^{1} = T^{2}$).

We now enumerate all possible realizations $X$ in $J \setminus Z$ of $T^{1} \oplus T^{2}$ that are consistent with the annotation of $\delta(e)$. Here, a consistent realization of $T^{1}$ is a subgraph of $J \setminus Z$ that is isomorphic to the subgraph of $H$ induced by the vertices of $H$ that correspond to the tokens in $T^{1}$, where additionally, if a token is assigned to the first (second) boundary by the annotation, then its realization has to be a vertex in the first (second) boundary, and if a token is assigned to the interior by the annotation, then its realization has to be a vertex in the interior of the strip.
A consistent realization of $T^{2}$ is defined similarly. Then a consistent realization of $T^{1} \oplus T^{2}$ is a consistent realization of $T^{1}$ and, if $T^{1} \not= T^{2}$, then also a consistent realization of $T^{2}$ containing only vertices that are distinct and not adjacent to the vertices of the realization of $T^{1}$.
If no consistent realization of $T^{1} \oplus T^{2}$ exists, then we blank the strip-edge, update $F'$ and $\delta$, and proceed to the next strip-edge. Using exhaustive enumeration, this procedure takes $n^{O(h)}$ time.

For each consistent realization $X$, we remove $N[X \cup Z]$ from $J$, and call the resulting graph $J'$.
From Theorem~\ref{thm:prelim:main}, we know that $\alpha(J') \leq 4$ or $J'$ is a fuzzy circular-arc graph.
Hence, it follows from Theorem~\ref{thm:algo:circ} and Proposition~\ref{prp:algo:alpha} that we can solve {\sc Induced Graph Matching} on~$J'$ in polynomial time.
Denote the solution for $X$ by $M^{e}(X)$.

Let~$X^{e}_{\max}$ denote a consistent realization $X$ such that $|M^{e}(X)|$ is maximum. 
Let $k' = \sum_{e} |M^{e}(X^{e}_{\max})|$, where the sum is over all strip-edges $e$ for which the corresponding edge elements have not been blanked. Intuitively, this means that we can find $k'$ occurrences of $H$ that are in the interiors of strips and that are conditioned on what the coloring says the boundary of the strips should look like.

\paragraph{Step~5:~Strip Global}
In this step, we find occurrences of $H$ that span multiple strips or affect a strip boundary.
To this end, we color the vertices of $G$ with a set of colors $\Gamma_{k}=\{0,\ldots,k\}$ as follows.
For each unblanked strip-edge $e$, let $(J,Z)$ denote the corresponding strip.
Recall that $X^{e}_{\max}$ consists of vertices that are realizations of tokens of at most two token groups, say token groups $i,j \in \{1,\ldots,k\}$.
We first color the vertices of $X^{e}_{\max}$ (in~$G$) by colors $i$ and $j$ from $\Gamma_{k}$ in the obvious way (i.e.~vertices from token group~$i$ receive color $i$), and color all remaining vertices of $V(J) \setminus Z$ with color $0$.
We then attempt to find one occurrence of $H$ in each color class of $\Gamma_{k} \setminus \{0\}$ independently.
This takes $n^{O(h)}$ time in total.
If this yields at least $k - k'$ occurrences of $H$, then we claim that $G$ has an induced $H$-matching of size $k$.
Otherwise, we proceed to the next coloring or base; if there are none left, then we return a ``no''-answer.

\medskip\noindent
The correctness of the algorithm is proved below.
  
\begin{theorem}
\label{thm:algo:main}%
{\sc Induced Graph Matching} is fixed-parameter tractable on claw-free graphs when parameterized by the size $k$ of the matching, for any fixed connected graph $H$.
\end{theorem}
\begin{proof}
To prove the theorem, we prove two claims:
\begin{enumerate} \renewcommand{\theenumi}{(\roman{enumi})}
\item If $k-k'$ occurrences of $H$ are found in Step~5, then $G$ has an induced $H$-matching of size $k$.
\item If $G$ has an induced $H$-matching of size $k$, then the algorithm will find one.
\end{enumerate}
The first claim ensures that if the algorithm returns something (that is not a ``no''-answer), then it does so correctly. The second claim speaks for itself. Together, the claims demonstrate the correctness of the algorithm described in this section.

\medskip
To prove Claim (i), it suffices to consider only Step~4 and Step~5. By the definition of a strip-structure, there are no edges between two vertices in the interior of different strips. Since $M^{e}(X^{e}_{\max})$ contains an induced $H$-matching in $J \setminus N[X^{e}_{\max} \cup Z]$, where $e$ is some unblanked strip-edge and $(J,Z)$ is the corresponding strip, the occurrences of $H$ in each of these induced $H$-matchings are independent. Hence, $K' := \bigcup_{e} M^{e}(X^{e}_{\max})$ is an induced $H$-matching of size $k'$, where the union is over all unblanked strip-edges $e$.
By construction, only vertices that belong to $X^{e}_{\max}$ will receive a color in $\Gamma_{k} \setminus\{0\}$. Since $M^{e}(X^{e}_{\max})$ does not contain any vertices from $J \setminus N[X^{e}_{\max} \cup Z]$ by construction, where again $e$ is some unblanked strip-edge and $(J,Z)$ is the corresponding strip, it follows that the occurrences of $H$ that are found in Step~5 are independent from those in $K'$.

It remains to show that the occurrences of $H$ found in Step~5 are independent of each other. Note that by Condition~2, for each vertex $b$ of $B$, the boundaries corresponding to $b$ of the edges incident to $b$ can contain tokens of at one most token group. Using Rule~2 and its third and fourth consequences, this means that for any strip-vertex $r$ of $\mathcal{R}$ all vertices of $C(r)$ receive a color in $\{0,j\}$, for some color $j \not= 0$. As the realizations of the two token groups $T^{1}$ and $T^{2}$ (if they are distinct) in a strip are independent, it follows that a vertex with color $i \in \Gamma_{k} \setminus\{0\}$ can only be adjacent to vertices of color $i$ or~$0$. Hence, the occurrences of $H$ found in Step~5 are independent of each other. This proves Claim (i).

\medskip
To prove Claim (ii), suppose that $G$ has an induced $H$-matching $M$ of size $k$. Consider the base promised by Proposition~\ref{prp:algo:base-k} and denote its hypergraph by $B$. By Proposition~\ref{prp:algo:cond1} and~\ref{prp:algo:cond2}, the base satisfies both Condition~1 and Condition~2. Using this base, we can define a natural coloring $\sigma$ using $\Gamma_{B}$ of the elements of $\mathcal{E}$ that correspond to the strip-vertices and strip-edges of the covered subgraph. By Proposition~\ref{prp:algo:colorcoding}, there is a coloring $F$ in $\mathcal{F}_{B}$ that contains $\sigma$ as a sub-coloring. We now only consider the steps of the algorithm for this choice of $B$ and $F$. Observe that after applying Rules~1 and 2, the resulting coloring $F'$ will still contain $\sigma$ as a sub-coloring, as by the third consequence of Rule~2, the colors of the strip-vertices determine the allowed colors for the incident strip-edges. Note that the surjection $\delta$ now contains (but is not necessarily equal to) a subgraph of $\mathcal{R}$ isomorphic to $B$.

In Step~4 of the algorithm, each unblanked strip-edge $e$ is treated individually, and therefore, we focus only on those strip-edges that are colored by $\sigma$. Let $e$ be such a strip-edge and $(J,Z)$ its corresponding strip. We know that $M$ yields a consistent realization $X^{e}_{M}$ of $T^{1} \oplus T^{2}$. Now observe that by Condition~1 and since $H$ is connected, if a token group has a token assigned to the interior of $\delta(e)$ but no tokens to the boundaries, then all tokens of the token group are assigned to the interior by the annotation of $\delta(e)$. We call these \emph{interior token groups}. We again know that $M$ contains an induced $H$-matching of $J \setminus N[X^{e}_{M} \cup Z]$ of size equal to the number of interior token groups of $e$. Then it follows that $k' = \sum_{e} |M^{e}(X^{e}_{\max})| \geq \sum_{e} |M^{e}(X^{e}_{M})|$ (where the first sum is over all unblanked strip-edges, and the second over all strip-edges that are colored by $\sigma$), which in turn is at least the number $t$ of interior token groups of the annotation.

In Step~5 of the algorithm, we note $\bigcup_{e} X^{e}_{M}$, where the union is over all strip-edges $e$ colored by $\sigma$, contains $k-t$ independent occurrences of $H$. Therefore, $\bigcup_{e} X^{e}_{\max}$ also contains $k-t$ independent occurrences of~$H$. In Step~5, the vertices of these occurrences will be given distinct colors of $\Gamma_{k}$, as the corresponding tokens belong to distinct token groups. It follows that Step~5 finds at least $k-t$ independent occurrences of $H$. As argued in the proof of Claim (i), this yields an induced $H$-matching of size $k$. This proves Claim~(ii).

\medskip
Finally, we need to show the running time. In Step~1, we need to enumerate all bases, which can be done in $(hk)^{O(hk)}$ time by Proposition~\ref{prp:algo:base}. Condition~1 and~2 can be checked in polynomial time. In Step~2, we need to enumerate all $\Gamma_{B}$-colorings of $\mathcal{F}_{B}$, which can be done in $g(k,h) \cdot n^{O(1)}$ time by Proposition~\ref{prp:algo:colorcoding}. The computations in Step~3--5 can each be performed in $n^{O(h)}$ time, using Theorem~\ref{thm:algo:circ} and Proposition~\ref{prp:algo:alpha}. Hence, the problem is indeed fixed-parameter tractable when parameterized by $k$ for any fixed connected graph $H$.
\qed\end{proof}

\section{Polynomial Kernel for Induced Graph Matching on Claw-Free Graphs when $H$ is Complete}
\label{sec:polynomialkernel}

The main result of this section is the following theorem.
\begin{theorem}
\label{thm:kernel:main}%
{\sc Induced Graph Matching} admits a polynomial kernel on claw-free graphs when parameterized by the size $k$ of the matching, for any fixed complete graph $H$.
\end{theorem}

Let $(G,H,k)$ be an instance of {\sc Induced Graph Matching}, where $G$ is a claw-free graph and $H$ is a complete graph. Let $n$ denote the number of vertices of $G$ and let $h$ denote the number of vertices of $H$.
Throughout, we assume that $\alpha(G) > 4$ and that $G$ is not a fuzzy-circular arc graph.
Otherwise, we can solve {\sc Induced Graph Matching} in polynomial time by Theorem~\ref{thm:algo:circ} and Proposition~\ref{prp:algo:alpha}, and reduce to a trivial ``yes''- or ``no''-instance.
We can then also assume that we have used Theorem~\ref{thm:prelim:main} and computed a strip-structure with a corresponding strip-graph~$\mathcal R$.

We start by showing how to bound the size of the strip-graph in Section~\ref{sec:poly:bound}, before giving the actual kernel in Section~\ref{sec:poly:kernel}.

\subsection{Bounding the Size of the Strip-Graph} \label{sec:poly:bound}
In this section, we show that we can assume that~$\mathcal R$ has a bounded number of strip-edges. We start by removing some vertices of $G$ that are irrelevant to the problem. The following proposition is trivial.

\begin{proposition}
\label{prp:kernel:easy-red}%
  If a vertex $v \in V(G)$ is not contained in any occurrence of $H$ in $G$, then $G$ has an induced $H$-matching of size $k$ if and only if $G - \{v\}$ has an induced $H$-matching of size $k$.
\end{proposition}
We show that some vertices can be removed even if they are contained in an occurrence of $H$ in $G$. For a strip-vertex $r \in V(\mathcal R)$, let $D_{r} = \bigcup_{e \in \mathcal R \mid r \in e} e \setminus\{r\}$ denote the set of distinct neighbors of $r$ in $\mathcal R$. We then define the \emph{dis-degree} of $r$ as the size of $D_{r}$, i.e.~this is the number of distinct neighbors of $r$ in~$\mathcal R$.

The idea behind the following lemma is reminiscent of arguments by Prieto and Sloper~\cite[Lemma~4]{PrietoSloper2006} and Dell and Marx~\cite[Observation~5.1]{DellMarx2012}.

\begin{lemma}
\label{lem:kernel:dis-degree}%
Let $r \in \mathcal R$ be a vertex in the strip-graph of dis-degree at least $2h(k-1) + h$.
Then $G$ has an induced $H$-matching of size $k$ if and only if $G' := G - C(r)$ has an induced $H$-matching of size $k-1$.
\end{lemma}
\begin{proof}
If $G$ has an induced $H$-matching $M$ of size $k$, then certainly $G'$ has an induced $H$-matching of size~$k-1$, since at most one occurrence of $H$ in~$M$ can include vertices of the clique $C(r)$.
For the converse direction, let~$M$ be an induced $H$-matching of $G'$ of size $k-1$.
Note that $M$ is also an induced $H$-matching of $G$.
Recall the definition of the covered subgraph of $\mathcal{R}$ with respect to $M$ from Definition~\ref{def:algo:covered}.
Let $M_{r} \subseteq D_{r}$ denote the set of strip-vertices contained in the covered subgraph.
By Proposition~\ref{prp:algo:covered}, $|M_{r}| \leq 2h(k-1)$.
For each strip-vertex $x$ of $D_{r} \setminus M_{r}$, pick a vertex of $C(r)$ that is contained in a strip corresponding to a strip-edge on $x,r$. Note that each such strip-edge is not covered.
Let $X$ denote the set of picked vertices.
As $|D_{r}| \geq 2h(k-1) + h$ by assumption, $|X| = |D_{r} \setminus M_{r}| \geq h$. Since $X$ consists of vertices of $C(r)$, $G[X]$ contains an occurrence of $H$. By picking vertices from strips corresponding to strip-edges that are not covered and whose strip-vertices are also not covered, we can guarantee that this occurrence of $H$ is independent of the occurrences of $H$ in $M$.
Hence, $G$ has an induced $H$-matching of size $k$.
\qed
\end{proof}

We now distinguish two types of strip-edges: those that are promising and those that are not. We call a strip-edge $e \in E(\mathcal R)$ \emph{promising} if it corresponds to a strip $(J_{e},Z_{e})$ such that $J \setminus N[Z_{e}]$ contains an occurrence of $H$; otherwise, the strip-edge is \emph{non-promising}. We will bound the number of strip-edges of both types, starting with promising strip-edges.

\begin{proposition}
\label{prp:kernel:promising}
If the strip-graph has at least $k$ promising strip-edges, then $G$ has an induced $H$-matching of size $k$.
\end{proposition}
\begin{proof}
Suppose that the strip-graph has at least $k$ promising strip-edges. Each such strip-edge corresponds to a strip $(J_{e},Z_{e})$ such that $J \setminus N[Z_{e}]$ contains an occurrence of $H$. Since such an occurrence of $H$ contains no vertices of $N[Z_{e}]$, it does not contain nor is it adjacent to vertices of any other strips. This implies that the set of occurrences of $H$ obtained by picking one occurrence of $H$ in $J \setminus N[Z_{e}]$ per promising strip-edge $e$ is an induced $H$-matching. Moreover, this induced $H$-matching contains at least $k$ occurrences by the assumption on the number of promising strip-edges.
\qed
\end{proof}
The proposition implies that we may assume that the strip-graph has at most $k-1$ promising strip-edges.

Next, we bound the number of non-promising strip-edges. We call a strip-edge \emph{helpful} if it corresponds to a stripe. Hence, strip-edges that correspond to spots are not helpful. Suppose that $\mathcal R$ contains $d > 2h$ non-promising strip-edges on $x,y$ for two strip-vertices $x,y \in V(\mathcal R)$ (where possibly $x=y$), and assume that $d'$ of those strip-edges are helpful. Among the non-promising strip-edges on $x,y$, select $\min\{d', 2h\}$ helpful ones and, if $d' < h$ additionally select $h-d'$ non-helpful ones. Now remove all vertices from the strips corresponding to strip-edges that were not selected. We call this a \emph{reduction step} on $x,y$.

\begin{lemma}
\label{lem:kernel:degree-bound}
Suppose that $\mathcal R$ contains $d > 2h$ non-promising strip-edges on $x,y$ for two strip-vertices $x,y \in V(\mathcal R)$ (where possibly $x=y$).
Let $G'$ be obtained from $G$ by performing a reduction step on $x,y$. Then $G'$ has an induced $H$-matching of size $k$ if and only if $G$ does.
\end{lemma}
\begin{proof}
One direction is trivial, so we focus on the other. Let $M$ be any induced $H$-matching of $G$ of size $k$.
Observe that any occurrence of $H$ in a vertex of a strip corresponding to a non-promising strip-edge on $x,y$ must have a vertex in $C(x)$ or $C(y)$.
Moreover, since $C(x)$ and $C(y)$ are cliques, neither can contain vertices of more than one occurrence of $H$ in~$M$.

Suppose that $x=y$. If no occurrence of $M$ contains a vertex of $C(x)$, then $M$ is also an induced $H$-matching of $G'$, and the lemma follows. So assume that an occurrence of $H$ in $M$ contains a vertex of $C(x)$. By the above observation, this occurrence is the only one in $M$ that could contain vertices of the strips corresponding to non-promising edges on $\{x\}$. We replace this occurrence of $H$ with one obtained by picking one vertex from the boundary of each of $h$ arbitrary strips whose corresponding strip-edges were selected in the reduction step. Since $C(x)$ is a clique, this indeed yields an occurrence of $H$. The resulting set of occurrences of $H$ is still an induced $H$-matching of $G$, but more importantly, it is also an induced $H$-matching of $G'$ of the same size, and the lemma follows.

Suppose that $x \not= y$ and that at least $2h$ non-promising strip-edges on $x,y$ are helpful.
Then the occurrence of $H$ in~$M$ containing a vertex of $C(x)$ (if it exists) can be replaced by an occurrence obtained by just taking $h$ vertices from $C(x)$, one vertex from the boundary of each stripe corresponding to the first $h$ strip-edges that were selected in the reduction step.
Simultaneously, the occurrence of $H$ in~$M$ containing a vertex of $C(y)$ (if it exists) can be replaced by an occurrence obtained by just taking $h$ vertices from $C(y)$, one vertex from the boundary of each stripe corresponding to the second $h$ strip-edges that were selected in the reduction step. The result of these replacements is still an induced $H$-matching of $G$, but is also one of $G'$ of the same size, and the lemma follows.

Suppose that $x \not= y$ and that less than $2h$ non-promising strip-edges on $x,y$ are helpful.
Since all helpful non-promising strip-edges are selected in the reduction step in this case, it suffices to consider the case that an occurrence of $H$ in $M$ contains a vertex of a strip corresponding to a non-helpful strip-edge, i.e.~to a spot. Observe that in this case only one occurrence of $H$ in $M$ can contain a vertex from a strip corresponding to a non-promising strip-edge on $x,y$. This occurrence can be replaced by an occurrence obtained by just taking $h$ vertices from $C(x)$, one vertex from the boundary of each stripe corresponding to the first $h$ strip-edges that were selected in the reduction step. The result of these replacements is still an induced $H$-matching of $G$, but is also one of $G'$ of the same size, and the lemma follows.
\qed
\end{proof}

\begin{lemma}
\label{lem:kernel:strip-bound}%
In $(hk)^{O(1)}\, n^{O(h)}$ time, we can either correctly decide whether $G$ has an induced $H$-matching of size $k$, or find an equivalent instance $(G',H,k')$ such that $G'$ is an induced subgraph of $G$ and $k' \leq k$, together with a strip-structure for $G'$ for which the strip-graph has $O(h^{4}k^{2})$ strip-edges.
\end{lemma}
\begin{proof}
We design an algorithm that iteratively reduces the size of $G$ and $k$, and along with it the number of strip-edges in the strip-graph of a strip-structure for $G$.

In the first step, we verify that each vertex of $G$ is contained in an occurrence of $H$; any other vertices are removed, which is safe by Proposition~\ref{prp:kernel:easy-red}. By exhaustive enumeration, this can be done in $h^{O(1)} \, n^{O(h)}$ time. If $G$ has more than one connected component, then we continue the algorithm on each connected component separately. We then verify that $G$ satisfies $\alpha(G) > 4$ and that $G$ is not a fuzzy-circular arc graph. This can be done in polynomial time by Theorem~\ref{thm:circ:fuzzy-reg}. Otherwise, we can solve {\sc Induced Graph Matching} in polynomial time by Theorem~\ref{thm:algo:circ} and Proposition~\ref{prp:algo:alpha}, and decide whether $G$ has an induced $H$-matching of size $k$ and stop the algorithm.

In the second step, we greedily find an inclusion-wise maximal induced $H$-matching $M$. This can be done in $h^{O(1)} \, n^{O(h)}$ time by iteratively finding an occurrence of $H$ in $G$ and then removing the occurrence and its neighbors. If $|M| \geq k$, then we can positively decide that $G$ has an induced $H$-matching of size $k$, and we stop the algorithm. Otherwise, we compute a strip-structure for $G$ in polynomial time using Theorem~\ref{thm:prelim:main}. Let $\mathcal R$ denote the corresponding strip-graph.

In the third step, we compute the dis-degree of each strip-vertex in $\mathcal R$. If the dis-degree of a strip-vertex $r$ exceeds $2h(k-1)+h$, then by Lemma~\ref{lem:kernel:dis-degree}, we can remove $C(r)$ from $G$ and reduce $k$ by $1$. We then return to the first step.

In the fourth step, we determine for each strip-edge whether it is promising. By exhaustive enumeration, this can be done in $h^{O(1)} \, n^{O(h)}$ time. If there are more than $k-1$ promising strip-edges, then by Proposition~\ref{prp:kernel:promising} we can positively decide that $G$ has an induced $H$-matching of size $k$, and we stop the algorithm.

In the fifth step, we consider two strip-vertices $x,y$ (where possibly $x=y$) such that the number of non-promising strip-edges on $x,y$ is greater than $2h$. If two such strip-vertices indeed exist, then we perform a reduction step on $x,y$, which reduces the graph $G$. We do not change $k$. By Lemma~\ref{lem:kernel:degree-bound}, we can then return to the first step.

If the algorithm gets past the fifth step, then we claim that $\mathcal R$ has $O(h^{4}k^{2})$ strip-edges. To see this, observe that the induced $H$-matching $M$ found in the second step has size at most $k-1$. Moreover, the dis-degree of each strip-vertex in $\mathcal R$ is at most $2h(k-1)+h$, there are at most $2h$ non-promising strip-edges on each pair of strip-vertices, and there are at most $k-1$ promising strip-edges.
By Proposition~\ref{prp:algo:covered}, the covered subgraph of $\mathcal R$ with respect to $M$ has at most $h(k-1)$ strip-edges and at most $2h(k-1)$ strip-vertices.
Therefore, there are at most $8h^{3}k^{2} + k$ strip-edges on a strip-vertex that is in the covered subgraph.
Now observe that by Proposition~\ref{prp:kernel:easy-red}, each strip-edge contains a vertex of an occurrence of $H$. Moreover, since $M$ is maximal, any occurrence of $H$ in $G$ must cover a strip-edge of~$\mathcal R$ incident to a strip-vertex of the covered subgraph. If $r \in V(\mathcal R)$ is a strip-vertex that is not in the covered subgraph, then $r$ can be contained in at most $h-1$ strip-edges, or it would be possible to extend $M$. Indeed, if $r$ is contained in at least $h$ strip-edges, then none of these strip-edges are covered by~$M$. Hence, $h$ arbitrary vertices of $C(r)$ would form an occurrence of $H$ that is independent of $M$, contradicting the maximality of~$M$. Since there can be at most $8h^{3}k^{2} + k$ strip-edges that are not in the covered subgraph, the strip-graph has at most $O(h^{4}k^{2})$ strip-edges.

To finish the proof of the lemma, note that each step of the algorithm takes $(hk)^{O(1)} \, n^{O(h)}$ time. Moreover, whenever the algorithm returns to its first step, the size of $G$ has been reduced. Hence, the algorithm runs in $(hk)^{O(1)} \, n^{O(h)}$ time.
\qed
\end{proof}

\subsection{A Polynomial Kernel} \label{sec:poly:kernel}
To obtain an intuition of the kernel, we give the following description of a simple algorithm for {\sc Induced Graph Matching} if $H$ is a fixed complete graph.
We first reduce the strip-graph~$\mathcal R$ using Lemma~\ref{lem:kernel:strip-bound}.
Then we observe that if we know the behavior of the induced $H$-matching on $C(r)$ for each $r \in \mathcal R$, then we can reduce to polynomial-time solvable instances of {\sc Induced Graph Matching} on individual strips.
Since the number of strip-vertices of~$\mathcal R$ is bounded by a function of $k$ for fixed $H$, this gives a fixed-parameter algorithm.
The kernel will mimic this algorithm by reducing to an instance of {\sc Weighted Independent Set} on general graphs of size polynomial in $k$.
By using a Karp-reduction, we can obtain an instance of {\sc Induced Graph Matching} on claw-free graphs again.

For our purposes, we define {\sc Weighted Independent Set} as the problem of given a graph $G'$, a weight function $w':V(G')\rightarrow\mathbb N$, and integers $k',K'\in\mathbb N$, to decide whether $G'$ has an independent set of size at least~$k'$ and weight at least $K'$. 

\begin{theorem}
\label{thm:kernel:crucial}%
In $n^{O(h)}$ time, we can reduce the instance $(G,H,k)$ to an equivalent instance of {\sc Weighted Independent Set} with $O( 2^{h} h^{4h+9} k^{2h+4} )$ vertices, maximum vertex-weight at most $k$, $k'$ bounded by $O(h^{4} k^{2})$, and $K'$ bounded by $k$.
\end{theorem}
\begin{proof}
Apply the algorithm of Lemma~\ref{lem:kernel:strip-bound}. If this algorithm decides that $G$ has an induced $H$-matching of size $k$ (or not), then return a trivial ``yes''-instance (respectively, ``no''-instance) of {\sc Weighted Independent Set} that satisfies the constraints of the theorem statement, and we are done. Otherwise, we consider the instance that Lemma~\ref{lem:kernel:strip-bound} produces and by abuse of notation denote it by $(G,H,k)$ as well. By the lemma, we may assume that $G$ has a strip-structure with a strip-graph $\mathcal R$ that has $O(h^{4}k^{2})$ strip-edges.

\paragraph{Construction Outline: Intuition \& Correctness}
We now construct an instance of {\sc Weighted Independent Set}. The basic idea is to create a  clique for each strip-vertex and for each strip-edge of the strip-graph; we call these cliques \emph{selection cliques}. Each vertex of the selection clique for a strip-vertex $x$ will correspond to a particular behavior of the induced $H$-matching on $C(x)$. Similarly, each vertex of the selection clique for a strip-edge $e$ will correspond to a particular behavior of the induced $H$-matching on the strip corresponding to $e$.
Observe that if there are no edges between the selection cliques, then an independent set contains exactly one vertex of each selection clique, meaning that it chooses exactly one behavior for each strip-vertex and for each strip-edge; we call this the \emph{selection property}.
However, without edges between the selection cliques, the behaviors selected by the independent set are not necessarily consistent, and thus might not point to an induced $H$-matching in $G$.  
Therefore, the crux of the construction is to add edges between the selection cliques to ensure consistent behaviors, while maintaining the selection property.

The difficulty in adding consistency edges is to limit the number of them that we add. To this end, we need more insight into the structure of an induced $H$-matching, and in particular of its covered subgraph. Since we assume $H$ to be a complete graph, we can indeed derive such a structure. If all strips would be stripes, then it follows from the definition of a strip-structure that each occurrence of $H$ is a subset of the set of vertices of a single strip or of $C(x)$ for a single strip-vertex $x$. The presence of spots, however, adds significantly to the complexity. Still, it follows from the definition of a strip-structure that the covered subgraph (after removing parallel strip-edges) of an occurrence of $H$ is isomorphic to a star or a triangle, as any two strip-edges have to be incident to a common strip-vertex for their vertices to be adjacent and thus to be part of the same occurrence of $H$. Moreover, if the simplified covered subgraph is isomorphic to a star, then either the occurrence is contained in $C(x)$ for a single strip-vertex $x$ (even though the number of strip-vertices in the covered subgraph can be $h+1$) or the occurrence covers at most two strip-vertices (and one strip-edge). 

The consequence of this structural observation is that the selection clique for a strip-vertex $x$ should not treat $x$ just as a singleton, but should also consider the behavior of $x$ in each pair and triple of strip-vertices. Therefore, the selection clique of each strip-vertex $x$ can be partitioned into three parts, containing vertices of type I, II, and III (plus several subtypes) to ensure consistency among singletons, pairs, and triples of strip-vertices, respectively. Note that the above structural observation also implies that these three types are sufficient to ensure consistency.

Finally, the construction assigns a weight to each vertex: this weight is equal to the number of occurrences of $H$ that the behavior corresponding to the vertex contributes. We then only need to bound the number of vertices in each selection clique, which combined with the bound on the number of strip-edges in $\mathcal{R}$ given by Lemma~\ref{lem:kernel:strip-bound}, yields a polynomial bound on the size of the construction.

Below, we describe the make-up of the selection cliques for the strip-edges and for the strip-vertices (per type), and simultaneously describe the consistency edges. We then prove that the construction has the properties that were promised in the theorem statement.

\paragraph{Selection Cliques for Strip-Edges}
We construct the selection clique for each strip-edge $e$ of $\mathcal R$. 
It is worthwhile to note the strong similarity between Step~4 of the algorithm of the previous section and the make-up of the selection cliques.
The make-up of the selection clique depends on $e$, and we distinguish three cases.

\emph{Spots:}\ In the first case, suppose that $e$ is a strip-edge on two strip-vertex $x,y$ such that the corresponding strip $(J_{e},Z_{e})$ is a spot. Recall that $V(J_{e}) \setminus Z_{e}$ is a singleton. An induced $H$-matching can either use this vertex in an occurrence of $H$ that contains vertices of $C(x)$ but not $C(y)$ (besides from spots between $x$ and $y$), of $C(y)$ but not $C(x)$ (besides from spots between $x$ and $y$), or of $C(x) \cup C(y)$, or it might not use this vertex at all. Therefore, the selection clique for $e$ has four vertices: $v^{e}_{x}$, $v^{e}_{y}$, $v^{e}_{xy}$, and $v^{e}_{0}$. These correspond to these four possible behaviors, respectively. All created vertices are assigned weight zero (we will account for their contribution elsewhere). 

\emph{Stripes on Two Strip-Vertices:}\ In the second case, suppose that $e$ is a strip-edge on two strip-vertices $x,y$ such that the corresponding strip $(J_{e},Z_{e})$ is a stripe. The behavior of an induced $H$-matching $M$ on $e$ is essentially determined by its behavior on $C(x)$ and $C(y)$, and in particular on the boundaries of $e$. Consider the boundary with respect to $x$, which is $N(z^{x}_{e})$. Then $M$ could contain an occurrence of $H$ that `sticks out' of the strip, meaning that uses a vertex of $N(z^{x}_{e})$, but not of $C(x) \setminus N(z^{x}_{e})$. On the other hand, $M$ could contain an occurrence of $H$ that `sticks in' to the strip, meaning that it uses a vertex of $C(x) \setminus N(z^{x}_{e})$ and between $1$ and $h-1$ vertices of $N(z^{x}_{e})$; note that no other occurrence of $H$ in $M$ can use a vertex of $N(z^{x}_{e})$. We can make a similar analysis for the boundary with respect to $y$. We also note that an occurrence of $H$ might stick in to both $x$ and $y$, or stick out of both $x$ and $y$. Finally, there might be no occurrence of $H$ in $M$ that uses a vertex of $C(x)$ or $C(y)$. After making a distinction according to the above analysis, the remaining occurrences of $H$ in $M$ that influence $e$ are those that only use vertices in the interior of $e$. By the properties of a strip-structure, these occurrences are completely independent of occurrences in other strips. Therefore, the selection clique should have a vertex for each case that we just distinguished, and correspond to a maximum induced $H$-matching that only uses vertices in the interior of $e$ and agrees with the case. To this end, for each $i,j \in \{-1,0,\ldots,h-1\}$ and $f \in \{I,C\}$, the selection clique for $e$ contains the vertex $v^{e}_{i,j,f}$ with weight equal to the size of a maximum induced $H$-matching on $J_{e} \setminus Z_{e}$ that:
\begin{itemize}
\item if $i=-1$ and $j=-1$, contains at least one vertex of $N(z^{x}_{e})$ and contains at least one vertex of $N(z^{y}_{e})$;
\item if $i=-1$ and $j \geq 0$, contains at least one vertex of $N(z^{x}_{e})$ and no vertices of $N[Y \cup \{z^{y}_{e}\}]$ for at least one set $Y \subseteq N(z^{y}_{e})$ of size $j$;
\item if $i \geq 0$ and $j=-1$, contains no vertices of $N[X \cup \{z^{x}_{e}\}]$ for at least one set $X \subseteq N(z^{x}_{e})$ of size $i$ and at least one vertex of $N(z^{y}_{e})$;
\item if $i \geq 0$ and $j \geq 0$, no vertices of $N[X \cup \{z^{x}_{e}\}]$ nor of $N[Y \cup \{z^{y}_{e}\}]$ for a pair of sets $X,Y$ for which $X \subseteq N(z^{x}_{e})$ and $|X| = i$, $Y \subseteq N(z^{y}_{e})$ and $|Y| = j$, and $X$ and $Y$ are independent (if $f = I$) or $X$ and $Y$ form a clique in $J_{e}$ (if $f = C$ and $i+j < h$).
\end{itemize}
If an induced $H$-matching with the given constraints does not exist, then the weight of the vertex is set to $-\infty$. Observe that $i=-1$ corresponds to having an occurrence of $H$ that sticks out of the boundary corresponding to $x$, and that $i \geq 0$ corresponds to having an occurrence of $H$ that sticks in to this boundary. A similar observation holds for $j=-1$ versus $j \geq 0$ with respect to the boundary corresponding to $y$. If $i,j \geq 0$, then we additionally need to know if the vertices of the occurrences that stick in to the boundaries are from the same occurrence, or from different occurrences: we use $f$ to determine this. Note that $f$ has no meaning for the first, second, and third items. In particular, in the first item it suffices to optimize over whether the occurrence(s) of $H$ using vertices of $N(z^{x}_{e})$ and $N(z^{y}_{e})$ are independent; for other strips, this information is irrelevant. It follows from Theorem~\ref{thm:algo:circ} and Proposition~\ref{prp:algo:alpha} that the weights can be computed in $n^{O(h)}$ time.

\emph{Stripes on One Strip-Vertex:}\ In the third case, suppose that $e$ is a strip-edge on one strip-vertex $x$. Then the corresponding strip $(J_{e},Z_{e})$ is a stripe. We note that here a similar analysis holds as in the case of a stripe on two strip-vertices. Therefore, for each $i\in \{-1,0,\ldots,h-1\}$, the selection clique for $e$ contains the vertex $v^{e}_{i}$ with weight equal to the size of a maximum induced $H$-matching on $J_{e} \setminus Z_{e}$ that:
\begin{itemize}
\item if $i=-1$, contains at least one vertex of $N(z^{x}_{e})$;
\item if $i \geq 0$, contains no vertices of $N[X \cup \{z^{x}_{e}\}]$ for at least one $X \subseteq N(z^{x}_{e})$ of size $i$.
\end{itemize}
If an induced $H$-matching with the given constraints does not exist, then the weight of the vertex is set to $-\infty$. It follows from Theorem~\ref{thm:algo:circ} and Proposition~\ref{prp:algo:alpha} that the weight can be computed in $n^{O(h)}$ time.

This describes the selection cliques for strip-edges. In the remainder, we will ignore the case that $e$ is a strip on one strip-vertex, as it is dealt with in a similar (but much more straightforward) manner as the case where is $e$ is on two strip-vertices.

\paragraph{Selection Cliques for Strip-Vertices -- Coordination}
It remains to coordinate the solutions of the strips. Due to the structural observation, each selection clique for a strip-vertex consists of three types of vertices (plus several subtypes): for each singleton, pair, and triple of strip-vertices. We consider each type in turn.

\paragraph{Type I}
Let $x$ be a strip-vertex and let $E_{x}$ denote the set of strip-edges that contain $x$. Vertices of type I of the selection clique for $x$ account for an occurrence of $H$ that is a subset of $C(x)$ and contains vertices of at least two strips (or for the case $C(x)$ contains no vertices of an occurrence of $H$), and for an occurrence of $H$ that contains a vertex in the boundary corresponding to $x$ of some strip whose strip-edge is in $E_{x}$ and contains no vertices in any other strip boundaries. This distinguishes two subtypes.

\emph{Type Ia:}\ 
For the first subtype, we consider an occurrence of $H$ that is a subset of $C(x)$ and contains vertices of at least two strips. We also consider the case that $C(x)$ contains no vertices of any occurrence of $H$. 
The crucial analysis or decision is which strips whose corresponding strip-edges are in $E_{x}$ contain a vertex of $H$, and if so, then how many. In other words, we need to determine the distribution of the vertices of $H$ over the strip-edges in $E_{x}$. To this end, let $\mathcal{U}$ denote the set of all sets $U$ of integers such that $\sum_{i \in U} i = h$, $\min_{i \in U} i \geq 1$, and $|U| \leq |E_{x}|$. We also add the set $\{0\}$ to $\mathcal{U}$.
Let $\mathcal{P}_{x}$ denote the set that contains for all $U \in \mathcal{U}$ all possible assignments of the numbers of $U$ to the strip-edges of $E_{x}$, where any strip-edge that is not assigned a number (which happens if $|U| < |E_{x}|$) is assigned $0$.
Observe that $\mathcal{P}_{x}$ contains all possible distributions of the vertices of $H$ over the strip-edges of $E_{x}$, plus an all-zero distribution to account for the case that no vertices of $C(x)$ are in an occurrence of $H$.
It is helpful to remove all $P$ from $\mathcal{P}_{x}$ for which $P$ assigns a number greater than one to a strip-edge that corresponds to a spot; note that no such distribution is possible as spots contain only one vertex of $G$.

The selection clique for $x$ now contains a vertex $v^{x}_{P}$ for each $P \in \mathcal{P}_{x}$. The weight of $v^{x}_{P}$ is set to one, unless $P$ assigns zero to every strip-edge, in which case the weight is set to zero. Since occurrences that stick in to a strip are not counted towards the weight of vertices in the selection clique of the strip, this setting of the weight is indeed correct.

It remains to add consistency edges. As expected, for a vertex $v^{x}_{P}$ for some $P \in \mathcal{P}_{x}$, we need to ensure that the distribution prescribed by $P$ is satisfied. 
Therefore, if $P$ assigns $i'$ to a strip-edge $e$ on $x,y$ that corresponds to a stripe for some strip-vertex $y$, then we make $v^{x}_{P}$ adjacent to:
\begin{enumerate}
\item $v^{e}_{i,j,I}$ for all $i\not=i'$ and all $j$;
\item $v^{e}_{i,j,C}$ for all $i$ and all $j$.
\end{enumerate}
Recall that since the occurrence of $H$ should form a subset of $C(x)$, we want to ensure that the sets $X$ and $Y$ in the definition of $v^{e}_{\cdot,\cdot,\cdot}$ for $i,j \geq 0$ are independent. Moreover, we want the strip corresponding to $e$ to reserve $i'$ vertices for use in the occurrence of $H$. These constraints are enforced by the above adjacencies.
If $P$ assigns $i'$ to a strip-edge $e$ on $x,y$ that corresponds to a spot for some strip-vertex $y$, then we make $v^{x}_{P}$ adjacent to:
\begin{itemize}
\item $v^{e}_{x}$, $v^{e}_{y}$, and $v^{e}_{xy}$ if $i' = 0$ and $P$ assigns a positive number to at least one strip-edge;
\item $v^{e}_{x}$ and $v^{e}_{xy}$ if $i' = 0$ and $P$ assigns zero to each strip-edge;
\item $v^{e}_{y}$, $v^{e}_{xy}$, and $v^{e}_{0}$ if $i' = 1$.
\end{itemize}
In the first case, the distribution $P$ says that we should select zero vertices from $e$ and at least one vertex from another strip-edge in $E_{x}$; therefore, only $v^{e}_{0}$ can be part of the independent set if $v^{x}_{P}$ is. In the second case, we should still account for an occurrence of $H$ that is a subset of $C(y)$, and therefore also $v^{e}_{y}$ should be allowed to be part of the independent set. The third case is immediate from the definition of $v^{e}_{\cdot}$ for a spot.

\emph{Type Ib:}\ 
For the second subtype, we consider an occurrence of $H$ that contains a vertex in the boundary corresponding to $x$ of some strip corresponding to a strip-edge $e \in E_{x}$ and that contains no vertices in any other strip boundaries. Therefore, for each strip-edge $e \in E_{x}$ that corresponds to a stripe, the selection clique for $x$ contains a vertex $v^{x}_{e}$ of weight zero. Since occurrences that stick out of a strip are counted towards the weight of vertices in the selection clique of the strip, this setting of the weight is indeed correct.

It remains to add consistency edges. Let $e$ be on $x,y$ for some strip-vertex $y$. Then we make $v^{x}_{e}$ adjacent to:
\begin{enumerate}
\item $v^{e}_{i,j,I}$ for all $i \not= -1$ and all $j$;
\item $v^{e}_{i,j,C}$ for all $i$ and all $j$;
\item $v^{e'}_{i,j,I}$ for all strip-edges $e' \in E_{x} \setminus\{e\}$ that correspond to a stripe, all $i \not= 0$, and all $j$;
\item $v^{e'}_{i,j,C}$ for all strip-edges $e' \in E_{x} \setminus\{e\}$ that correspond to a stripe, all $i$, and all $j$;
\item $v^{e'}_{x}$, $v^{e'}_{y}$, and $v^{e'}_{xy}$ for all strip-edges $e' \in E_{x} \setminus\{e\}$ that correspond to a spot.
\end{enumerate}
The motivation for these adjacencies is the same as it would be in the first subtype for a distribution $P$ that assigns zero to all strip-edges in $E_{x} \setminus\{e\}$ and $-1$ to $e$.

\paragraph{Type II}
Let $x,y$ be two strip-vertices for which there is a strip-edge on $x,y$, and let $E_{xy}$ denote the set of strip-edges on $x,y$. Vertices of type II of the selection cliques for $x$ and for $y$ account for an occurrence of $H$ that contains vertices of both $C(x)$ and $C(y)$ and either sticks in to or sticks out of a stripe whose corresponding strip-edge is in $E_{xy}$. Accordingly, we distinguish two subtypes.

\emph{Type IIa:}\ 
The first subtype will express that an occurrence of $H$ is contained in $C(x) \cup C(y)$ and sticks in to a stripe whose corresponding strip-edge is in $E_{xy}$. Note that from the definition of a strip-structure, there is at most one stripe which this occurrence can stick in to. We also note that if an induced $H$-matching contains this occurrence, then the spots among $E_{xy}$ can only contain vertices of this occurrence, and not of any other occurrence. In fact, it is always best to use as many spots as possible to realize the occurrence.
Let $\ell$ be the number of strip-edges in $E_{xy}$ that correspond to a spot. By the reduction steps performed as part of Lemma~\ref{lem:kernel:strip-bound}, we know that $\ell \leq h$. If $\ell = 0$ or $\ell \geq h-1$, then we create no vertices: in the first case, the occurrence will only use vertices from stripes and this is properly handled by vertices of type Ib and IIb; in the second case, the occurrence should only use vertices from $C(x)$ or $C(y)$ and this is properly handled by vertices of type Ia. Hence, $0 < \ell < h-1$. 

For each strip-edge $e \in E_{xy}$ that corresponds to a stripe, the selection clique for $x$ contains the vertex $v^{xy,e}_{x}$ and the selection clique for $y$ contains the vertex $v^{xy,e}_{y}$. These vertices will express that the occurrence of $H$ contains vertices from both boundaries of the stripe corresponding to $e$ and from the $\ell$ spots.
We set the weight of $v^{xy,e}_{x}$ to one and the weight of $v^{xy,e}_{y}$ to zero. 
Since occurrences that stick in to a strip are not counted towards the weight of vertices in the selection clique of the strip, we indeed need that the weight of these vertices sums to one, and the setting of the weights is thus correct.

It remains to add consistency edges. We make $v^{xy,e}_{x}$ adjacent to:
\begin{enumerate}
\item $v^{e'}_{x}$, $v^{e'}_{y}$, and $v^{e'}_{0}$ for all strip-edges $e'$ on $x,y$ that correspond to a spot;
\item $v^{e}_{i,j,C}$ for all $i,j \geq 0$ that do not sum to $h-\ell$, for $i=0$ and $j=h-\ell$, for $i=h-\ell$ and $j=0$, for $i = -1$ and all $j$, and for all $i$ and $j=-1$;
\item $v^{e}_{i,j,I}$ for all $i$ and all $j$;
\item $v^{e'}_{i,j,f}$ for all $i,j$ that are not both equal to $0$, all $f \in \{I,C\}$, and all strip-edges $e' \in E_{xy} \setminus\{e\}$ that correspond to a stripe;
\item $v^{e'}_{x'}$, $v^{e'}_{y'}$, and $v^{e'}_{x'y'}$ for all strip-edges $e' \in E_{x'y'}$ that correspond to a spot and are on some strip-vertices $x',y'$ such that $|\{x',y'\} \cap \{x,y\}| = 1$;
\item $v^{e'}_{i,j,f}$ for all $i\not=0$, all $j$, all $f$, and all strip-edges $e'$that contain exactly one of $x,y$ and that correspond to a stripe, where $i$ is the index that corresponds to $x$ or $y$;
\item all vertices in the selection clique for $y$ except $v^{xy,e}_{y}$.
\end{enumerate}
The first adjacency guarantees that we take all vertices from spots that correspond to strip-edges in $E_{xy}$. The second adjacency guarantees that the same occurrence sticks in to $e$ and uses at least one vertex from each boundary (if the occurrence would use a vertex from just a single boundary, then it would be in $C(x)$ or $C(y)$, and thus handled by vertices of type Ia). Moreover, we do not care how the $h - \ell$ vertices of the occurrences of $H$ that stick in to $e$ are distributed over the two boundaries of $e$; we leave this optimization to the independent set problem, as is evident from the second adjacency. 
The second and third adjacencies together ensure that the same instance sticks into both boundaries.
The fourth adjacency guarantees that we select no vertices from the boundaries of stripes that correspond to a strip-edge in $E_{xy} \setminus\{e\}$. The fifth and sixth adjacencies guarantee that we select no vertices of $C(x)$ or $C(y)$ from the boundaries of strips corresponding to strip-edges not on $x,y$ but containing $x$ or $y$. The final adjacency ensures that we select $v^{xy,e}_{x}$ if and only if we select $v^{xy,e}_{y}$.

The consistency edges for $v^{xy,e}_{y}$ are similar.

\emph{Type IIb:}\ 
The second subtype will express that one or two occurrences of $H$ contain vertices of both $C(x)$ and $C(y)$ and stick out of a stripe that corresponds to a strip-edge in $E_{xy}$. Note that in this case, none of the spots participate. This type in a way complements type Ib and IIa. For each strip-edge $e$ on $x,y$ that corresponds to a stripe, the selection clique for $x$ contains the vertex $\bar{v}^{xy,e}_{x}$ of weight zero and the selection clique for $y$ contains the vertex $\bar{v}^{xy,e}_{y}$ of weight zero.
Since occurrences that stick out of a strip are counted towards the weight of vertices in the selection clique of the strip, we indeed need that the weight of these vertices sums to zero, and the setting of the weights is thus correct.

It remains to add consistency edges. We make $\bar{v}^{xy,e}_{x}$ adjacent to:
\begin{enumerate}
\item $v^{e'}_{x}$, $v^{e'}_{y}$, and $v^{e'}_{xy}$ for all strip-edges $e' \in E_{xy}$ that correspond to a spot;
\item $v^{e}_{i,j,f}$ for all $i,j$ that are not both equal to $-1$ and all $f \in \{I,C\}$.
\item $v^{e'}_{i,j,f}$ for all $i,j$ that are not both equal to $0$, all $f \in \{I,C\}$, and all strip-edges $e' \in E_{xy} \setminus\{e\}$ that correspond to a stripe;
\item $v^{e'}_{x'}$, $v^{e'}_{y'}$, and $v^{e'}_{x'y'}$ for all strip-edges $e' \in E_{x'y'}$ that correspond to a spot and are on some strip-vertices $x',y'$ such that $|\{x',y'\} \cap \{x,y\}| = 1$;
\item $v^{e'}_{i,j,f}$ for all $i\not=0$, all $j$, all $f$, and all strip-edges $e'$ that contain exactly one of $x,y$ and that correspond to a stripe, where $i$ is the index that corresponds to $x$ or $y$;
\item all vertices in the selection clique for $y$ except $\bar{v}^{xy,e}_{y}$.
\end{enumerate}
The first adjacency guarantees that we do not take vertices from spots that correspond to strip-edges in $E_{xy}$. The second adjacency guarantees that we select one or two occurrences of $H$ that stick out of $e$. The motivation for the other adjacencies is the same as in type IIa.
 
The consistency edges for $\bar{v}^{xy,e}_{y}$ are similar.

\paragraph{Type III}
Vertices of type III account for an occurrence of $H$ that contains a vertex of each of $C(w)$, $C(x)$, and $C(y)$ for three strip-vertices $w,x,y$. It follows from the definition of a strip-structure that this is only possible if the occurrence only uses vertices from spots. Moreover, the occurrence contains at least one vertex from a spot that corresponds to a strip-edge in $E_{wx}$, at least one vertex from a spot that corresponds to a strip-edge in $E_{wy}$, and at least one vertex from a spot that corresponds to a strip-edge in $E_{xy}$. Observe also that all vertices from spots that correspond to a strip-edge in $E_{wx} \cup E_{wy} \cup E_{xy}$ either contribute to the occurrence or, in any induced $H$-matching that contains the occurrence, are not adjacent to any other occurrence in the induced $H$-matching.
Therefore, let $w,x,y$ be three strip-vertices such that there is a strip-edge on each pair $w,x$, and $w,y$, and $x,y$ that corresponds to a spot. Moreover, the number $\ell$ of strip-edges on $w,x$, and $w,y$, and $x,y$ that correspond to a spot has to be at least $h$. 
Then the selection clique for $w$ contains the vertex $v^{wxy}_{w}$ of weight one, the selection clique for $x$ contains the vertex $v^{wxy}_{x}$ of weight zero, and the selection clique for $y$ contains the vertex $v^{wxy}_{y}$ of weight zero.
Since occurrences that stick in to a strip are not counted towards the weight of vertices in the selection clique of the strip, we indeed need that the weight of these vertices sums to one, and the setting of the weights is thus correct.

It remains to add the consistency edges. We make $v^{wxy}_{w}$ adjacent to:
\begin{enumerate}
\item $v^{e}_{x'}$, $v^{e}_{y'}$, and $v^{e'}_{0}$ for each strip-edge $e$ on $x',y'$ that corresponds to a spot, where $e \in E_{wx} \cup E_{wy} \cup E_{xy}$;
\item $v^{e'}_{i,j,f}$ for all $i,j$ that are not both equal to $0$, all $f \in \{I,C\}$, and all strip-edges $e' \in E_{wx} \cup E_{wy} \cup E_{xy}$ that correspond to a stripe;
\item $v^{e'}_{x'}$, $v^{e'}_{y'}$, and $v^{e'}_{x'y'}$ for all strip-edges $e' \in E_{x'y'}$ that correspond to a spot and are on some strip-vertices $x',y'$ such that $|\{x',y'\} \cap \{w,x,y\}| = 1$;
\item $v^{e'}_{i,j,f}$ for all $i\not=0$, all $j$, all $f \in \{I,C\}$, and all strip-edges $e'$ that contain exactly one of $w,x,y$ and that correspond to a stripe, where $i$ is the index that corresponds to $w$, $x$, or $y$;
\item all vertices in the selection clique of $x$ except $v^{wxy}_{x}$ and all vertices in the selection clique of $y$ except $v^{wxy}_{y}$.
\end{enumerate}
The first adjacency guarantees that we select all vertices of spots that correspond to a strip-edge in $E_{wx} \cup E_{wy} \cup E_{xy}$. Note that this possibly selects too many vertices, but as observed before, this does not create any problems. The motivation for the other adjacencies is similar as in type IIa and IIb.

The consistency edges for $v^{wxy}_{x}$ and $v^{wxy}_{y}$ are similar.

\paragraph{Final Construction}
Call the constructed graph $G'$. To finish the instance of {\sc Weighted Independent Set}, we set $k'$ equal to the number of strip-vertices and strip-edges of $\mathcal R$ and $K' = k$.

\paragraph{Properties}
We now prove each of the properties of the reduction that are claimed in the theorem statement. Let $N$ denote the number of strip-edges. Recall that, by Lemma~\ref{lem:kernel:strip-bound}, $N = O(h^{4}k^{2})$.

\emph{Weights:}\ 
We can bound the weights as follows. If there is a vertex of weight $k$ or more, then it must correspond to a stripe. By construction, this stripe has an induced $H$-matching of size at least $k$, and we can reduce to a trivial ``yes''-instance of {\sc Weighted Independent Set} instead. If a vertex has weight $-\infty$, then we can simply remove it from $G'$. Therefore, all vertex weights are between $0$ and $k-1$.

\emph{Size:}\ 
To bound the size of $G'$, we first bound the size of $\mathcal{P}_{x}$ for each strip-vertex $x$ (recall vertices of type Ia).  Observe that $|\mathcal{U}| = O(2^{h})$ (more precise bounds are known, but this will suffice).
We can assume that the numbers of each $U \in \mathcal{U}$ are ordered from small to large.
To bound $|\mathcal{P}_{x}|$, we note that the number of ordered subsets of $|E_{x}|$ of size at most $\min\{h,|E_{x}|\}$ is at most $|E_{x}|^{h}$.
For each such ordered subset $X$, we create a bijection to each $U \in \mathcal{U}$ with $|U| = |X|$ by assigning the $i$-th element of $U$ to the $i$-th element of $X$.
We then extend this assignment to $E_{x}$ by assigning zero to any elements of $E_{x} \setminus X$.
This shows that $|\mathcal{P}_{x}|$ is $O( 2^{h} N^{h} )$.

We are now ready to bound the size of $G'$.
For each strip-edge, we created a selection clique of at most $2(h+2)^{2}$ vertices. Hence, there are $O(Nh^{2})$ such vertices.
For each strip-vertex, instead of bounding the size of each selection clique directly, it is easier to analyze how many vertices were created of each type.
We created $O( 2^{h} N^{h} )$ vertices of type Ia and at most $N$ of type Ib for each strip-vertex.
We created at most $2N$ vertices of type IIa and at most $2N$ of type IIb for each pair of strip-vertices.
We created at most $3$ vertices for each triple of strip-vertices.
In total, this means that $O( 2^{h} h^{4h+8} k^{2h+4} )$ vertices were created.

\emph{Time:}\ 
Note that the bound on the size of $G'$ together with Theorem~\ref{thm:algo:circ}, Proposition~\ref{prp:algo:alpha}, and the trivial assumption that $h,k \leq n$ implies that $G'$ can be computed in $n^{O(h)}$ time.

\emph{Correctness:}\ 
We claim that $(G,H,k)$ is a ``yes''-instance for {\sc Induced Graph Matching} if and only if $(G',w',k',K')$ is a ``yes''-instance for {\sc Weighted Independent Set}.
For the one direction, consider an induced $H$-matching $M$ of size $k$.
Now construct an independent set $I'$ of $G'$ as follows.
First, we determine the vertex that $I'$ will contain in the selection clique for each strip-edge. 
Consider the case that $e$ is a strip-edge on strip-vertices $x,y$ that corresponds to a stripe $(J_{e},Z_{e})$. Let $M_{e}$ denote the set of occurrences of $H$ in $M$ that cover $e$.
If there are two occurrences of $H$ in $M_{e}$ that stick out of $e$ or one occurrences that sticks out of both boundaries of $e$, then select $v^{e}_{-1,-1,I}$ (or $v^{e}_{-1,-1,C}$, it does not matter). 
If there is one occurrence of $H$ in $M_{e}$ that sticks in to both boundaries of $e$, then select $v^{e}_{i,j,C}$, where $i,j$ are the number of vertices of $N(z^{x}_{e})$ and $N(z^{y}_{e})$, respectively, that are in this occurrence.
We now treat all remaining cases as follows. If there is an occurrence of $H$ in $M_{e}$ that sticks out of the boundary that corresponds to $x$, set $i = -1$. If there is an occurrence of $H$ in $M_{e}$ that sticks in to the boundary that corresponds to $x$, set $i$ equal to the number of vertices in $N(z^{x}_{e})$ of this occurrence. If there is no occurrence that uses a vertex of $N(z^{x}_{e})$, set $i=0$. Do the same for $j$ with respect to $y$. Then select $v^{e}_{i,j,I}$. Observe that these cases are exhaustive, and thus we always select an appropriate vertex in the selection clique of $e$. The case that $e$ is a strip-edge on one strip-vertex is similar. In the case that $e$ is a strip-edge on strip-vertices $x,y$ that corresponds to a spot, we select the appropriate vertex as discussed in the paragraph on the selection clique for spots.

Second, we select the vertex that $I'$ will contain in the selection clique for each strip-vertex. Let $x$ be a strip-vertex.
If there is no occurrence of $H$ in $M$ that covers $x$, then select $v^{x}_{P}$, where $P$ is the all-zero distribution.
So consider an occurrence of $H$ in $M$ that covers $x$.
If this occurrence is a subset of $C(x)$ and covers at least two strip-edges, then select $v^{x}_{P}$, where $P$ is the distribution of the occurrence over $E_{x}$.
If this occurrence covers only one strip-edge $e$ and the only strip-vertex that it covers is $x$, then this occurrence sticks out of $e$ but not out of the other boundary of $e$, and we select $v^{x}_{e}$.
If this occurrence is not a subset of $C(x)$ and covers $x$, some other strip-vertex $y$, and a strip-edge $e$ that corresponds to a stripe, then this occurrence sticks in to $e$ or sticks out of $e$. In the former case, select $v^{xy,e}_{x}$. In the latter case, select $\bar{v}^{xy,e}_{x}$. 
Finally, if this occurrence covers strip-edges on three distinct pairs of strip-vertices where two include $x$, then select $v^{wxy}_{x}$, where $w,y$ are the other two strip-vertices that are involved.
By the structural observation in the construction outline, these cases are exhaustive, and thus we always select an appropriate vertex in the selection clique of $x$.

By the construction of the consistency edges, none of the selected vertices are adjacent. Therefore, the constructed set $I'$ is an independent set. Since $I'$ contains a vertex in the selection clique of each strip-vertex and strip-edge, $I'$ has size $k'$. Finally, by the construction of the weights, $I'$ has weight $K'$.

For the converse, let $I'$ be an independent set of $G'$ of at least size $k'$ and at least weight $K'$. Then $I'$ contains one vertex in the selection clique of each strip-vertex and each strip-edge (and thus $I'$ in fact has size exactly $k'$). The vertices in the selection clique prescribe which vertices in $G$ should selected, as discussed in the construction outline. By the construction of the consistency edges, these behaviors are consistent. Since the weight of $I'$ is at least $K'$ and by the construction of the weights, the resulting induced $H$-matching has size at least $k$.
\qed
\end{proof}

We are now ready to give the actual kernel.

\begin{proof}[of Theorem~\ref{thm:kernel:main}]
  If $|V(H)| = 1$, then the problem is actually {\sc Independent Set}.
  Since this can be solved in $O(n^{3})$ time on claw-free graphs~\cite{FaenzaOrioloStauffer11}, we can reduce to a trivial ``yes''- or ``no''-instance.
  So assume that $|V(H)| > 1$.
  Since {\sc Weighted Independent Set} is in $\mathsf{NP}$ and {\sc Induced Graph Matching} on claw-free graphs is $\mathsf{NP}$-hard when $H$ is a fixed complete graph of size at least two~\cite{GareyJ1979,KirkpatrickH1983,KoblerR2003}, there exists a Karp-reduction from {\sc Weighted Independent Set} to {\sc Induced Graph Matching} on claw-free graphs when $H$ is a complete graph.
  Note that when the input graph $G'$ to {\sc Weighted Independent Set} is a graph with size polynomial in $k$ and maximum vertex-weight at most $k$, the {\sc Induced Graph Matching} instance produced by this reduction has size polynomial in $k$.
  The proof now follows from Theorem~\ref{thm:kernel:crucial}.
\qed
\end{proof}

\section{Parameterized Intractability on $K_{1,4}$-Free Graphs}
\label{sec:hardness}
In this section, we prove that {\sc Induced Graph Matching} is $\mathsf{W}[1]$-hard with respect to $k$ on $K_{1,4}$-free graphs when $H$ is any fixed complete graph.
For this, we first show that {\sc Independent Set} is $\mathsf{W}[1]$-hard on $K_{1,4}$-free graphs, by presenting a reduction from the {\sc Multicolored Clique} problem.
In this problem, we are given a graph $G$ and an assignment of colors $\{1,\ldots,k\}$ to the vertices of the graph, and the question is whether $G$ admits a $k$-clique using exactly one vertex of each color.
This problem is known to be $\mathsf{W}[1]$-hard~\cite{FellowsHRV09} parameterized by $k$.

\begin{theorem} \label{thm:hardness:is}
  {\sc Independent Set} on $K_{1,4}$-free graphs, parameterized by the size $k$ of the independent set, is $\mathsf{W}[1]$-hard.
\end{theorem}
\begin{proof}
  Let $(G,k)$ be an instance of {\sc Multicolored Clique} where each vertex has been assigned a color in $\{1,\ldots,k\}$.
  Let $V_{i}$ denote the set of vertices in $V(G)$ having color $i$, and assume that they are ordered $v_{i}^{1},\ldots,v_{i}^{|V_{i}|}$.
  For each color $i$, construct a gadget on $|V_{i}| \cdot (k-1)$ vertices.
  We may imagine them to be organized in a matrix-like shape with $|V_{i}|$ rows and $k-1$ columns, such that each row corresponds to a vertex of $V_{i}$ and each column to a color other than $i$.
  We then add edges between $(v_{i}^{p},j)$ and $(v_{i}^{q},j)$ for each $p\not=q$ with $1 \leq p,q \leq |V_{i}|$ and each $j \in \{1,\ldots,k\}\backslash\{i\}$.
  This makes each column of the matrix a clique.
  We also add an edges between $(v_{i}^{p},j)$ and $(v_{i}^{q},j+1)$ for each $1 \leq q < p \leq |V_{i}|$ and each $j \in \{1,\ldots,k\}\backslash\{i\}$, where we take $j+1$ to be $1$ if $j = k$.

  For each pair of colors $i,j$ ($i \not= j$), we also construct a matrix-like gadget.
  It has $|E_{ij}|$ rows, where $E_{ij} = (V_{i} \times V_{j}) \cap E$, and two columns, one for color~$i$ and one for color $j$.
  Each element of the matrix-like gadget again consists of one vertex.
  We add all edges between the vertices of the gadget, except that the edge between $(e,i)$ and $(e,j)$ is not present for any $e \in E_{ij}$.
  In particular, this makes each column of the matrix a clique.

  We then connect the vertex and edge gadgets as follows.
  Consider a pair of colors $i,j$ and some $u \in V_{i}$.
  We then connect $(v_{i}^{p},j)$ in the gadget for $V_{i}$ to all $(e,i)$ in the gadget for $E_{ij}$ for which $v_{i}^{p} \not\in e$.
  Call the resulting graph $G'$.
  To complete the instance $(G',k')$, we set $k' = k(k-1) + 2 \binom{k}{2} = 2k(k-1)$.

  We first prove that $G'$ is $K_{1,4}$-free.
  Suppose not, and consider an occurrence of $K_{1,4}$ in $G'$.
  Call the vertex of degree four of this occurrence the center of the occurrence and the other vertices the leafs.
  Suppose that the center is $(v_{i}^{p},j)$ for some $1 \leq i \leq k$, $1 \leq p \leq |V_{i}|$, and $j \in \{1,\ldots,k\}\backslash\{i\}$.
  By construction, $(v_{i}^{p},j)$ is only adjacent to vertices in the same column of the $V_{i}$ vertex gadget and in the columns $j-1$ and $j+1$, and to vertices in the $i$ column of the $E_{ij}$ edge gadget. Since each of these columns is a clique, each column contains exactly one leaf of the occurrence. Consider the leafs of the occurrence in the $V_{i}$ vertex gadget and let them be $(v_{i}^{q},j-1)$, $(v_{i}^{r}, j)$, and $(v_{i}^{s}, j+1)$. By construction, $q < p < s$. Suppose that $r > q$. Then by construction $(v_{i}^{q},j-1)$ and $(v_{i}^{r}, j)$ are adjacent, a contradiction. Hence $r \leq q$. Then by construction $(v_{i}^{r}, j)$ and $(v_{i}^{s}, j+1)$ are adjacent, a contradiction. Hence the center is $(e,i)$ for some $1 \leq i,j \leq k$ with $i\not= j$ and $e \in E_{ij}$. By construction, $(e,i)$ is only adjacent to vertices in the $j$ column of the $V_{i}$ vertex gadget, and to vertices in the $i$ or $j$ column of $E_{ij}$. Since each of these columns are a clique, the neighborhood of $(e,i)$ can be partitioned into three disjoint cliques. This contradicts that $(e,i)$ has an independent set of size four in its neighborhood. Hence, $G'$ is indeed $K_{1,4}$-free.

  It remains to show that $G$ has a multicolored clique $K$ of size $k$ if and only if $G'$ has an independent set $I$ of size $k'$. To this end, we prove the following claim.

  \medskip
  \noindent\textbf{Claim:} Any maximum independent set of the $E_{ij}$ edge gadget has size $2$ and consists of a single row of the gadget.
    Any maximum independent set of the $V_{i}$ vertex gadget has size $k-1$ and consists of a single row of the gadget.\\
  \textbf{Proof of Claim:} The first part is immediate from the construction of $E_{ij}$.
    For the second part, recall that each column of the gadget is a clique, and thus contains at most one vertex from any independent set.
    Since each row is an independent set, the size of any maximum independent set is $k-1$.
    Also note that if $(v_{i}^{p},j)$ is part of a maximum independent set, then the only vertices of the form $(v_{i}^{q},j+1)$ that can be part of this maximum independent set are those for which $p < q$. Then any maximum independent set consists of the vertices of only a single row of the matrix, and has size $k-1$.
    This completes the proof of the claim.\hfill\#

  \medskip
  Suppose that $G$ has a multicolored clique $K$ of size $k$.
  We construct the following independent set $I$.
  For the unique vertex $v_{i}^{p}$ in $K \cap V_{i}$ for each color~$i$, we add the $k-1$ vertices on the row of the gadget of~$V_{i}$ that corresponds to $v_{i}^{p}$ to $I$. For the unique edge $e$ in $K \cap E_{ij}$ for each pair of colors $i,j$, we add the two vertices on the row of the gadget of $V_{i}$ that corresponds to $e$ to~$I$.
  It is immediate from the description of the construction that this is an independent set and that it has size $k'$.

  Suppose that $G'$ has an independent set $I$ of size $k'$.
  By the above claim, each edge gadget has $2$ vertices in $I$ that are in the same row.
  Similarly, each vertex gadget has $k-1$ vertices in $I$ that are in the same row.
  Now observe that the selected vertex rows and edge rows together induce a multicolored clique.
  This clique has size $k$ by our choice of $k'$.
\qed
\end{proof}

We now show how this theorem implies the hardness of {\sc Induced Graph Matching} on $K_{1,4}$-free graphs for any fixed complete graph $H$.
\begin{theorem}
\label{thm:hardness:clique}
  {\sc Induced Graph Matching} on $K_{1,4}$-free graphs, parameterized by the size $k$ of the sought induced $H$-matching, is $\mathsf{W}[1]$-hard for any fixed complete graph $H$.
\end{theorem}
\begin{proof}
  Let $(G,k)$ be an instance of {\sc Independent Set} on $K_{1,4}$-free graphs.
  Construct a graph $G'$ from $G$ by replacing each vertex $v \in V(G)$ with a clique $K_v$ of size $|V(H)|$, and making all vertices of $K_v$ adjacent to all vertices of $N(v)$ in $G$.
  That is, we replace $v$ by a twin set of size $|V(H)|$.
  Call this twin set $H_{v}$.
  Clearly,~$G'$ has an induced $K_{1,4}$ if and only if $G$ does.
  The resulting instance of {\sc Induced Graph Matching} is $(G',H,k)$.

  Suppose that $G$ has an independent set $I$ of size $k$.
  Let $M$ be the set of $H_{v}$ for all $v \in I$.
  Since $I$ is an independent set, $M$ must be an induced $H$-matching.
  It clearly has size $k$.

  Suppose that $(G',H,k)$ has an induced $H$-matching $M$ of size $k$.
  Consider a copy $H'$ of $H$ in $M$ which uses vertices from $H_{v}$ for any $v$ in some set $V_{H'} \subseteq V(G)$ with $|V_{H'}| \geq 2$.
  Then, by construction of $G'$, no copy of $H$ in~$M$ can use a vertex of $H_{u}$ for any $u \in V(G)\setminus V_{H'}$.
  Moreover, no copy of $H$ in $M$ other than~$H'$ can use a vertex of $H_{u}$ for any $u \in V(H')$.
  But then we can just replace $H'$ by $H_{v}$ for some $v \in V_{H'}$.
  It follows that $\{v \mid H_{v} \in M \}$ is an independent set of $G$ of size $k$.

  The theorem is now immediate from Theorem~\ref{thm:hardness:is}.
\qed
\end{proof}

We observe that similar reductions can be given for graphs $H$ that are not a complete.
Then, however, we can no longer ensure that the constructed graph is $K_{1,4}$-free; rather, the excluded induced subgraphs of the resulting graph will depend on the structure of $H$.

\section{Discussion}
\label{sec:discussion}
In this section, we discuss the possibility of extending the work of this paper.
In particular, we give some open problems.

We showed that {\sc Induced Graph Matching} on $K_{1,3}$-free graphs is fixed-parameter tractable for fixed connected graphs $H$, and even admits a polynomial kernel if $H$ is a complete graph.
It is natural to ask whether these results extend to $K_{1,4}$-free graphs.
Theorem~\ref{thm:hardness:clique}, however, shows that {\sc Induced Graph Matching} becomes $\mathsf{W}[1]$-hard on such graphs.
We also note that the running times of our results contain a term $n^{O(h)}$, where $h = |V(H)|$.
Unfortunately, it has been shown that extending to the case where $h$ is also a parameter is unlikely to succeed, as {\sc Induced Graph Matching} becomes $\mathsf{W}[1]$-hard even on line graphs and co-bipartite graphs~\cite{GolovachPL2012,GolovachPL12}.

It is unclear, however, whether the fixed-parameter algorithms in this paper for {\sc Induced Graph Matching} extend to graphs $H$ with multiple components.
This seems to be difficult even for fuzzy circular-arc graphs.
For the same reason, it is open whether parameterization by $k$ is necessary in Theorem~\ref{thm:circ:circarc-circarc-fpt}, i.e.~whether {\sc Induced Graph Matching} is fixed-parameter tractable on proper interval graphs when parameterized by $|V(H)|$.
In both of these cases, the natural approach suggested by the results in Section~\ref{sec:circ} and by Cameron and Hell~\cite{CameronHell2006} (to reduce to an {\sc Independent Set} problem) seems to fail.
For the second case, even if we are given all occurrences of $H$ in $G$, then following the described approach we would need to solve an instance of {\sc Independent Set} on a multi-interval graph, which is known to be $\mathsf{W}[1]$-hard~\cite{FellowsHRV09}.
For the first case, the graph on which we would need to solve {\sc Independent Set} is even more complex.

The extension of Theorem~\ref{thm:circ:circarc-circarc-poly} and~\ref{thm:circ:circarc-circarc-fpt} from long proper circular-arc graphs to long circular-arc graphs should also be considered open.
The proofs of these theorems imply a reduction from {\sc Induced Graph Matching} where $H$ is a proper circular-arc graph and $G$ is a long circular-arc graph to {\sc Induced Subgraph Isomorphism} on such graphs.
However, the correctness of the reported polynomial-time algorithm for this problem~\cite{HeggernesMV2010} on the more restrictive interval graphs $G$ cannot be verified at this time~\cite{HeggernesMV2012}.

Another interesting open question is whether the polynomial kernel for {\sc Induced Graph Matching} on claw-free graphs extends beyond the case when $H$ is a fixed complete graph.
The proof of Theorem~\ref{thm:kernel:main} actually extends to this more general case.
However, the obstacle to extend our current techniques is that we are unable to find a lemma similar to Lemma~\ref{lem:kernel:strip-bound} if $H$ is not a complete graph.

\begin{acknowledgements}
  We thank the anonymous reviewers for helpful remarks improving the presentation of this manuscript.
\end{acknowledgements}

\appendix\normalsize

\section{\textsc{Induced Graph Matching} on Triangles is $\mathsf{NP}$-hard on Line Graphs} \label{sec:k3}
Recall from the introduction that the immediate correspondence between an $H$-matching in a graph $G$ and an induced $L(H)$-matching in $L(G)$ does not apply in case $L(H)$ is a triangle, i.e.~in case that $H$ is $K_{3}$ or $K_{1,3}$.
Hence, to the best of our knowledge, the complexity of {\sc Induced Graph Matching} for $H = K_3$ (also known as {\sc Induced Triangle Matching} or {\sc Induced Triangle Packing}) on line graphs is open.

In this section, we prove that {\sc Induced Graph Matching} for $H = K_3$ on line graphs is $\mathsf{NP}$-hard.

\begin{theorem} \label{thm:line:hard}
{\sc Induced Graph Matching} for $H = K_3$ is $\mathsf{NP}$-complete on line graphs of planar graphs.
\end{theorem}
\begin{proof}
It is clear that {\sc Induced Graph Matching} for $H = K_3$ on line graphs of planar graphs belongs to $\mathsf{NP}$.
To prove $\mathsf{NP}$-hardness, we reduce from {\sc Independent Set} on planar graphs where each vertex has degree exactly three, which is known to be $\mathsf{NP}$-complete~\cite{MaierS1977}.
Let $(G,k)$ be an instance of this problem. We transform this to an instance of {\sc Induced Graph Matching} with $H = K_{3}$ as follows. First, we subdivide each edge of $G$. That is, we (simultaneously) remove each edge $e = \{u,v\}$ and add a new vertex $x_{e}$ and new edges $\{u,x_{e}\}$ and $\{x_{e},v\}$. Denote the resulting graph by $G'$. Then the instance of {\sc Induced Graph Matching} is $(L(G'), K_{3}, k)$, where $L(G')$ is the line graph of $G'$. Note that $G'$ is planar, and thus $L(G')$ is the line graph of a planar graph.

Observe that any induced subgraph in $L(G')$ that is isomorphic to $K_{3}$ corresponds to three edges of $G'$ that are incident on the same vertex. Armed with this observation, we show that $(L(G'), K_{3}, k')$ is a ``yes''-instance of {\sc Induced Graph Matching} if and only if $(G,k)$ is a ``yes''-instance of {\sc Independent Set}, thus completing the proof.

Suppose that $(G,k)$ is a ``yes''-instance of {\sc Independent Set}. Let $I$ be an independent set of $G$ of size $k$. For each $v \in I$, consider the three edges of $G'$ that are incident on $v$. These correspond to a triangle in $L(G')$. Let $M$ be the set of these triangles for all vertices in $I$. Note that $|M| = k$. By the above observation, no two triangles in $L(G')$ have a vertex in common. Moreover, since $I$ is independent, no two triangles in $M$ are adjacent. Therefore, $(L(G'), K_{3}, k)$ is a ``yes''-instance of {\sc Induced Graph Matching}.

Suppose that $(L(G'), K_{3}, k)$ is a ``yes''-instance of {\sc Induced Graph Matching}. Let $M$ be an induced $K_{3}$-matching in $L(G')$ of size $k$. By the above observation, no two triangles in $L(G')$ have a vertex in common, and by construction each triangle corresponds to some vertex $v$ that is both in $G$ and $G'$. Let $I$ be the set of vertices that correspond to the triangles of $M$. Since $M$ cannot contain triangles that correspond to adjacent vertices in $G$, $I$ is an independent set. As $|I|=|M|=k$, $(G,k)$ is a ``yes''-instance of {\sc Independent Set}.
\qed\end{proof}

\end{document}